%% file: main_draft.tex
\documentclass[12pt,a4paper]{article}
\usepackage[truedimen,margin=30mm]{geometry} 

\usepackage{mathrsfs}
\usepackage{amssymb}
\usepackage{amsmath}
\usepackage{ascmac}
\usepackage{amsthm}
\usepackage{graphicx}
\usepackage{natbib}
\usepackage{setspace}
\usepackage{csquotes}
\usepackage{times}
\usepackage{bm}
\usepackage{color}

\usepackage{graphicx}
\usepackage{rotating} 

\usepackage{algorithm}
\usepackage{algpseudocode}

\usepackage{booktabs}
\usepackage{multirow}

\usepackage{titlesec}
\titleformat*{\section}{\large\bfseries}
\titleformat*{\subsection}{\it}

\newtheorem{thm}{Theorem}
\newtheorem{lem}{Lemma}

\newtheorem{prp}{Proposition}

\newtheorem{as}{Assumption}
\def\pr{{\mathrm{pr}}}

\def\bmu{{\bm{\mu}}}
\def\bSigma{{\bm{\Sigma}}}
\def\bC{{\bm{C}}}
\def\bCyp{{\bm{C}_{y'}}}
\def\br{{\bm{r}}}
\def\bS{{\bm{S}}}
\def\tr{{\mathrm{tr}}}
\def\by{{\bm{y}}}

\def\bz{{\bm{z}}}
\def\bth{{\bm{\theta}}}
\def\etab{{\bm{\eta}}}
\def\beps{{\bm{\epsilon}}}
\def\bxi{{\bm{\xi}}}
\def\bze{{\bm{\zeta}}}

\newcommand{\var}[1]{\mathrm{var}\left(#1\right)}
\newcommand{\E}[1]{E\left(#1\right)}

\title{{\bf On Misspecified Error Distributions in Bayesian Functional Clustering: Consequences and Remedies}}

\date{}

\begin{document}

\maketitle
\doublespacing

\vspace{-1.5cm}
\begin{center}
{\large Fumiya Iwashige$^1$, Tomoya Wakayama$^{2,\ast}$, Shonosuke Sugasawa$^3$, Shintaro Hashimoto$^1$}

\medskip

\medskip
\noindent
$^1$Department of Mathematics, Hiroshima University\\
$^2$Center for Advanced Intelligence Project, RIKEN\\
$^3$Faculty of Economics, Keio University
\end{center}

\begingroup
\renewcommand\thefootnote{\fnsymbol{footnote}}
\footnotetext[1]{Corresponding author, Email: tomoya.wakayama@riken.jp.}
\endgroup

\vspace{1mm}
\begin{center}
{\bf \large Abstract}
\end{center}

Nonparametric Bayesian approaches provide a flexible framework for clustering without pre-specifying the number of groups, yet they are well known to overestimate the number of clusters, especially for functional data. We show that a fundamental cause of this phenomenon lies in misspecification of the error structure: errors are conventionally assumed to be independent across observed points in Bayesian functional models. Through high-dimensional clustering theory, we demonstrate that ignoring the underlying correlation leads to excess clusters regardless of the flexibility of prior distributions. Guided by this theory, we propose incorporating the underlying correlation structures via Gaussian processes and also present its scalable approximation with principled hyperparameter selection. Numerical experiments illustrate that even simple clustering based on Dirichlet processes performs well once error dependence is properly modeled.


\newpage
\input{Introduction}

\input{BayesianFunctionalClustering}

\input{Theory}

\input{Proposal}

\input{Simulation}

\input{Discussion}

\section*{Acknowledgements}
This work is partially supported by the Japan Society for the Promotion of Science (JSPS) KAKENHI grant numbers 24K21420, 25H00546 and 25K07131,  the Japan Science and Technology Agency (JST) ACT-X grant number JPMJAX23CS, and JST SPRING grant number JPMJSP2132.

\vspace{1cm}
\bibliographystyle{chicago}
\bibliography{ref}

\newpage
\input{Supplemet}

\end{document}

%% file: Introduction.tex
\section{Introduction}

Bayesian functional clustering has become an essential approach for analyzing heterogeneous functional data arising in diverse scientific fields such as biostatistics, neuroscience, and econometrics. 
Applications include clustering of gene-expression profiles \citep{do2006bayesian}, growth trajectories in longitudinal studies \citep{james2003clustering}, and electroencephalographic signals \citep{SuarezGhosal2016_BayesianClusteringLocalFeatures,margaritella2023bayesian}. 
A recent review on functional clustering is provided by \cite{zhang2023review}.
In such contexts, functional observations often exhibit complex structures that cannot be captured by a single homogeneous model, motivating the use of clustering techniques to group functions with similar patterns. 
Nonparametric Bayesian methods, particularly those based on the Dirichlet process~\citep[DP,][]{ferguson1973bayesian,ferguson1974prior}, are widely adopted because they simultaneously infer the number of clusters and cluster-specific parameters without the need for pre-specification \citep{Hjort_Holmes_Müller_Walker_2010,Muller2015,ghosal2017fundamentals}.

Despite their popularity, clustering methods with Dirichlet processes are well known to suffer from a tendency to overestimate the number of clusters \citep{miller2013simple,miller2018mixture}. This problem is particularly severe in functional clustering, frequently producing a large number of singleton groups \citep{rodriguez2014functional}.
To mitigate this problem, more flexible partition priors than Dirichlet processes have been proposed. For instance, \citet{rodriguez2014functional} employed the generalized Dirichlet process and \cite{rigon2023enriched} introduced a partition distribution using functional constraints, both of which aim to reduce the prior probability of partitions with many clusters to prevent degenerate allocations dominated by singletons. While existing studies have primarily focused on modifying partition priors, we argue that the more fundamental issue lies in the specification of the error distribution. In most Bayesian functional models, observational errors are assumed to be independent across discrete sampling points \citep{petrone2009hybrid,rodriguez2014functional,suarez2017bayesian}. 
This assumption is rarely justifiable in practice, especially when functions are densely observed and strong local correlations in measurement noise naturally arise \citep{ramsay2005functional,Gabrys01092010,BENHENNI201980,Gertheiss2024}. 
Moreover, functional data are known to exhibit the phenomenon of perfect classification, whereby certain covariance and mean structures can render individual trajectories almost perfectly separable \citep{delaigle2012achieving,wakayama2024fast}. 
This implies that high-dimensional functional observations carry very strong individual information, so that under the independence assumption, the accumulation of information across points may artificially force each curve into its own cluster, resulting in many singletons.
If such dependencies are ignored, posterior clustering probabilities can be distorted, leading to an overestimation of the number of clusters. 
Our conjecture is that even when flexible priors are used to penalize partitions with many clusters, the independence assumption on errors can still drive clustering towards excessive partitions. Conversely, if the correlation structure among error terms is properly modeled, clustering performance can be substantially improved even with the standard Dirichlet process prior.

The objective of this paper is twofold. 
First, we show that the above conjecture is true via theoretical analysis of how misspecified error distributional assumptions affect posterior clustering.
Specifically, we employ an asymptotic framework with an increasing number of grid points (with sample size fixed or growing slowly), which is different from the large-sample regimes for existing studies on the number of clusters of nonparametric Bayesian models for non-functional observations \citep[e.g.][]{miller2013simple,miller2014inconsistency,10.1093/biomet/asac051,Alamichel2024bayesian}.
Under the framework, we show that ignoring the error dependence leads to a larger number of clusters depending on the decay rate of the eigenvalues of the correlation function in the error term.
Second, based on the theory, we propose a practical modeling strategy that incorporates correlation in the error distribution to address the issue of misspecification.
In particular, we provide an approximate and scalable modeling approach that uses Gaussian processes and satisfies the theoretical requirements of the error distributions.
Our experiments suggest that in functional clustering, misspecification of the error structure, especially the ubiquitous independent error assumption, can be a more immediate driver of over-clustering than the partition prior itself, and appropriately modeling the error distribution is an effective remedy even with a simple Dirichlet process prior. The code is available in the publicly accessible GitHub repository: https://github.com/TomWaka/DPM-with-banded-GP-for-FDA.

This paper is organized as follows.
Section~2 introduces the problem setup and Bayesian functional clustering models.
Section~3 presents a theoretical analysis of how error misspecification affects clustering performance.
In Section~4, we develop a practical approach to address the issue discussed in Section~3. Numerical experiments are provided in Section~5. Directions for future work are presented in Section~6.
All technical proofs and details of the posterior computation algorithm are provided in the Supplementary Material.

%% file: BayesianFunctionalClustering.tex
\section{Bayesian Functional Clustering}

We formulate Bayesian functional clustering models by combining Gaussian processes with a Dirichlet process mixture model. 
Let $y_1(x),\ldots,y_n(x)$ represent $\mathbb{R}$-valued random functions on a compact domain $\mathcal{X} \subset \mathbb{R}$. We consider the following data-generating process: 
\begin{align}\label{mixture-function}
    y_i(x)\mid\mu_i(x)\sim {\rm GP}(\mu_i(x),C_y),\ \ \ \ \  \mu_i(x)= \sum_{k=1}^{K^*} \theta_k(x)I(z_i=k), 
\end{align}
for $i=1,\ldots,n$, where ${\rm GP}(\mu,C)$ denotes a Gaussian process with mean function $\mu:\mathcal{X} \to \mathbb{R}$ and covariance function $C(\cdot, \cdot):\mathcal{X}\times \mathcal{X}\to \mathbb{R}$~\citep{RasmussenWilliams2006}, $K^*$ is the unknown number of clusters, and $\theta_k(\cdot)$ is the common mean function within the $k$th cluster.
Here, $z_1,\ldots,z_n$ are unknown membership variables, and partition distributions induced from (generalized) Dirichlet process or Pitman-Yor process~\citep{pitman1997two,ishwaran2001gibbs} are typically introduced to perform posterior inference on $z_1,\ldots,z_n$. For $\theta_k(x)$, we also assume a Gaussian process, ${\rm GP}(\eta(x),C_\theta)$. Note that the mean functions and covariance functions are continuous and square integrable throughout the paper. In both the theoretical analysis and numerical experiments presented in this paper, we assume that $\eta(x) \equiv 0$. Our goal is to cluster the functions $\{\mu_i(x)\}$ into an unknown number of groups based on their latent structure, where each cluster corresponds to a distinct functional pattern modeled by a Gaussian process.

In practice, instead of the random function itself, we obtain discretized observations.
For simplicity, we assume that the function values are observed at $m$ points, denoted by $x_1,\ldots,x_m$, for all $n$ functions.  
Let $\by_i=(y_i(x_1),\ldots,y_i(x_m))$ be the $m$-dimensional vector of discretized observations. 
Then, at the data analysis stage, we work with the following model for $\by_i$:
\begin{align}\label{mixture}
&\by_i\mid \bmu_i\sim \mathcal{N}(\bmu_i,\bC_y),\ \ \  \bmu_i= \sum_{k=1}^{K_{n,m}} \bm{\theta}_k I(z_i=k), \ \ \ 
\bm{\theta}_k \sim  \mathcal{N}(\etab,\bC_{\theta}),
\end{align}
where $\mathcal{N}(\bmu, \bSigma)$ is the multivariate Gaussian distribution with mean vector $\bmu$ and variance-covariance matrix $\bSigma$, $\bmu_i=(\mu_i(x_1),\ldots,\mu_i(x_m))$, $\bth_k=(\theta_k(x_1),\ldots,\theta_k(x_m))$, $\etab=(\eta(x_1),\ldots,\eta(x_m))$, and $\bC_y$ and $\bC_{\theta}$ are $m\times m$ variance-covariance matrices. 
Here, $K_{n,m}$ is the number of distinct values (or clusters) among the $n$ observations and is an $\mathbb{N}$-valued random variable whose distribution is determined by specifying a partition process. Although the notation $K_n$ is commonly used~\citep{ghosal2017fundamentals}, we adopt $K_{n,m}$ in this paper to reflect the setting in which the number of grid points $m$ increases. For $\bz_{1:n} = (z_1,\ldots,z_n)$, we assume a joint distribution, expressed as the Chinese restaurant process, denoted by ${\rm CRP}(\alpha)$, with the concentration hyperparameter $\alpha>0$ that controls the number of clusters~\citep{Muller2015,ghosal2017fundamentals}. This prior process assigns probabilities to existing clusters in proportion to their cluster sizes and assigns a probability of creating a new cluster in proportion to $\alpha$.

Inference can be performed using Markov chain Monte Carlo (MCMC) methods, such as Gibbs sampling with collapsed or marginalized updates. 
In the MCMC algorithm, each label variable $z_i$, for $i=1,\dots,n$, is updated as follows:
\begin{align}\label{z_update}
    \pr(z_{i} = k\mid\bz_{-i},\bth_{1:K_{-i}},\by_{1:n}) \propto
    \begin{cases}
    n_{-i,k}\phi(\by_i\mid\bth_k,\bC_y) & \text{if } k=1,\ldots,K_{-i}  \\
    \alpha \phi(\by_i\mid \etab, \bC_y+\bC_\theta),  & \text{if } k=K_{-i}+1
  \end{cases},
\end{align}
where $\bz_{-i}$ denotes the vector obtained from $\bz$ by removing the $i$th element, $K_{-i}$ the number of unique values in $\bm{z}_{-i}$, $n_{-i,k}:=\#\{j\neq i:z_j=k\}$ the number of elements in $\bm{z}_{-i}$ equal to $k$, and $\phi(\cdot \mid \bmu,\bSigma)$ the probability density function of $\mathcal{N}(\bmu, \bSigma)$.

We refer to the model obtained by replacing $\bC_y$ with $\bCyp$ in \eqref{mixture} as \emph{the assumed model}. 
In practice, since $\bC_y$ in the true data-generating process is unknown, the assumed model specifies a variance-covariance structure determined by the statistician. In particular, when $\bCyp = \sigma^2 \bm{I}_m$, the assumed model is referred to as \emph{the independent-error model}. Equation \eqref{z_update} clearly illustrates that the specification of the error model exerts a direct influence on the resulting cluster assignments produced by the algorithm.

In Bayesian functional clustering, especially when the number of grid points is large, it has been widely observed that misspecification of the error model induces an overestimation of the number of clusters. A common assumption in many existing studies is the use of an independent-error model, not only for its simplicity but also due to its computational advantages \citep[see, e.g.,][]{petrone2009hybrid,rodriguez2014functional,suarez2017bayesian}. 
However, with a large number of observed points, the impact of misspecified error distribution tends to be amplified, which can lead to significant distortion in clustering outcomes.
While some previous work has addressed this issue by adopting more flexible partition models, we argue that the fundamental cause lies not in the partition structure itself but in the misspecification of the error model. 
Although the impact of misspecified error distributions on functional clustering has been widely observed in empirical studies \citep[e.g.][]{zhu2024clustering}, it has received little attention from a theoretical perspective, and few solutions have been proposed. In the following section, we take a closer look at this issue.

%% file: Theory.tex
\section{Consequences of Misspecified Error Distributions
}

In this section, we theoretically show that when the assumed model's errors are independent, over-clustering tends to arise. We then characterize conditions under which over-clustering is unlikely when the assumed model exhibits dependent errors. 

We introduce the notation used throughout this section. For both the true model and the assumed model, we denote by $\pr^{\mathrm{true}}$ and $\pr^{\ast}$ the probabilities of updating each label $z_i$ in \eqref{z_update}:
\begin{align}
    \pr^{\mathrm{true}}(z_{i} = k\mid\bz_{-i},\bth_{1:K_{-i}},\by_{1:n}) &\propto
    \begin{cases}
    n_{-i,k}\phi(\by_i\mid\bth_k,\bC_y) & \text{if } k=1,\ldots,K_{-i}  \\
    \alpha \phi(\by_i\mid \bm{0}, \bC_y+\bC_\theta)  & \text{if } k=K_{-i}+1
  \end{cases}\label{z_update_true}, \\[1em]
  \pr^{\ast}(z_{i} = k\mid\bz_{-i},\bth_{1:K_{-i}},\by_{1:n}) &\propto
    \begin{cases}
    n_{-i,k}\phi(\by_i\mid\bth_k,\bCyp) & \text{if } k=1,\ldots,K_{-i}  \\
    \alpha \phi(\by_i\mid \bm{0}, \bCyp+\bC_\theta)  & \text{if } k=K_{-i}+1
  \end{cases}\label{z_update_assum},
\end{align}
where we can set $\etab=\bm{0}$ in \eqref{mixture} without loss of generality.

\subsection{Independent-error model}
We show that the independent-error model has a higher probability of generating a new cluster than the true model. To this end, we first provide some notation and assumptions. 

Let $\mathbb{S}^{m}_{++}$ be the set of all $m \times m$ positive definite symmetric matrices. 
For $\bm{A} \in \mathbb{S}^{m}_{++}$, let $\lambda_1(\bm{A}),\dots,\lambda_m(\bm{A})$ be the eigenvalues of $\bm{A}$, ordered such that $\lambda_1(\bm{A}) \geq \dots \geq \lambda_m(\bm{A})$. We impose the following assumptions on $\bC_y$ and $\bC_\theta$ in \eqref{mixture}.
\begin{as}[Eigenvalue decay]\label{as:ED}
    For $\bC_y$ and $\bC_\theta$, the sequences of eigenvalues satisfy $\lambda_m(\bC_y)\to 0$ and $\lambda_m(\bC_\theta)\to 0$ as $m\to\infty$.
\end{as}
Assumption~\ref{as:ED} requires that the eigenvalues of the covariance kernels of Gaussian processes converge to zero. This assumption requires the functions to be smooth. From Mercer's expansion, it follows that rough processes require many eigenfunctions, which leads to slower decay of the eigenvalues (\cite{RasmussenWilliams2006}, Section 4.3). Many standard stationary kernels satisfy this assumption. Specifically, the Mat\'ern kernel $k_{\text{Mat\'ern}}(d)= \tau^{2}\,(2^{1-\nu}) (\sqrt{2\nu}\, d/\ell)^{\nu}K_{\nu}(\sqrt{2\nu}\, d/\ell) /\Gamma(\nu)$ with distance $d = \|x -x'\|_2$ and scale $\tau^2>0$, length-scale $\ell>0$, and smoothness $\nu>0$, satisfies Assumption~\ref{as:ED} \citep{stein1999interpolation, RasmussenWilliams2006}, where $K_{\nu}$ denotes the modified Bessel function of the second kind. The Gaussian (squared-exponential) and exponential kernels are special cases of the Mat\'ern kernel.

We use standard Landau notation throughout. For functions $f,g:\mathbb{N}\to\mathbb{R}$ such that $g(n)\neq0$ for sufficiently large $n$, we write $f(n)=\omega(g(n))$ if $|f(n)/g(n)|\to\infty$ as $n\to\infty$ and $f(n)=o(g(n))$ if $|f(n)/g(n)|\to0$ as $n\to\infty$.
\begin{as}[Determinant and trace conditions]\label{as:DT}
    For $\bC_y$ and $\bC_\theta$, it holds that, for some $0\leq \gamma < 1/2$,
    \[
    \log\left\{\frac{\det(\bC_y+\bC_\theta)}{\det\bC_y}\right\} = \omega (m^{1-\gamma}),\quad \tr\left\{(\bC_\theta+\sigma^2\bm{I}_m)^{-1}(\bC_\theta+\bC_y)\right\} -\tr(\bC_y) = o(m^{1/2}).
    \]
\end{as}

Assumption~\ref{as:DT} implies that the variance structure is dominated by $\bC_\bth$, while the $\bC_y$ plays only a subsidiary role. The determinant restriction corresponds to a signal-to-noise ratio condition on the eigenvalues. For example, if $\bC_y$ and $\bC_\bth$ are simultaneously diagonalizable, the left side is $\sum_{j=1}^m \log\{1 + \lambda_j(\bC_\bth) / \lambda_j(\bC_y)\}$. The fact that this sum diverges at a polynomial rate with exponent greater than $1/2$ indicates that the signal $\bC_\bth$ dominates the noise $\bC_y$. The trace condition also concerns the relative contribution of the eigenvalues of the noise $\bC_y$ compared to those of the signal $\bC_\bth$. This condition controls the magnitude of the (whitened) $\bC_y$ and implies that the true error noise does not exhibit excessively strong long-range dependence or an overly large scale.

For the independent-error model, the following theorem holds under the assumptions.
\begin{thm}\label{thm:nonconsistent}
    Suppose that $\bC_y$ and $\bC_\theta$ in \eqref{mixture} satisfy Assumptions~\ref{as:ED} and \ref{as:DT}, every $\bth_k$ follows $\mathcal{N}(\bm{0},\bC_{\theta})$ and $\bCyp = \sigma^2 \bm{I}_m$. Then, for fixed finite $n$ and any $\bz_{-i}$, the following holds:
    \[
    \frac{\pr^{\ast} \left(z_i=K_{-i}+1
             \mid \bz_{-i},
                  \bth_{1:K_{-i}},
                  \by_{1:n}\right)}
         {\pr^{\mathrm{true}} \left(z_i=K_{-i}+1
             \mid \bz_{-i},
                  \bth_{1:K_{-i}},
                  \by_{1:n}\right)}
     \xrightarrow{ p } \infty,
    \qquad m\to\infty.
    \]
\end{thm}

Theorem~\ref{thm:nonconsistent} shows that, under Assumptions~\ref{as:ED} and~\ref{as:DT} and as $m\to\infty$, the independent error model assigns a higher probability to generating a new cluster in the MCMC update than the true model. Intuitively, by ignoring the true correlations, the assumed covariance $\bCyp=\sigma^2 \bm{I}_m$ is smaller than $\bC_y$, making the likelihood for existing clusters $\phi(\by_i\mid\bth_k,\bCyp)$ overly concentrated around its mean and thus it is evaluated to be smaller. Consequently, the relative weight of the new-cluster term $\alpha\phi(\by_i\mid\bm0,\bCyp+\bC_\theta)$ increases, which leads to more frequent creation of new clusters.

The result suggests that in this misspecified setting, the posterior distribution of the number of clusters, $K_{n,m}$, induced by the independent-error model tends to place substantial mass on large integers ($\leq n$). For finite-dimensional data, with $n\to\infty$, it is known that Dirichlet process mixtures are inconsistent for estimating the number of clusters \citep{miller2013simple, miller2014inconsistency}. Our addition is to show, in the functional data setting with $m\to\infty$, that the independent-error misspecification leads to overestimation of the number of clusters relative to a well-specified (dependent-error) model.

To better specify when Assumptions~\ref{as:ED} and~\ref{as:DT} hold, we provide the following proposition. For functions $f,g:\mathbb{N}\to\mathbb{R}$ such that $g(n) \neq 0$ for sufficiently large $n$, we write $f(n) \gtrsim g(n)$ if there exists a constant $c > 0$ such that $c \leq |f(n)/g(n)|$ for sufficiently large $n$.
\begin{prp}\label{prp:nonconsistent}
Fix $\nu>0$ and $\kappa>0$ and assume that the eigenvalues of
$\bC_y,\bC_\theta\in\mathbb S_{++}^m$ obey $\lambda_j(\bC_y)=j^{-2\nu-1}$ and $
\lambda_j(\bC_\theta)=j^{-\kappa}$ for $j=1,2,\dots ,m$.
Let $L=\log\left\{\det(\bm{I}_m + \bS_y)\right\}$ with $\bS_y=\bC_y^{-1/2}\bC_\theta \bC_y^{-1/2}$. Then, $L$ diverges if $\kappa<2\nu+2$, and its rate is
\begin{align*}
L\gtrsim 
\begin{cases}
m\log m, & \kappa<2\nu+1,\\
m, & \kappa=2\nu+1,\\
m^{2\nu+2-\kappa}, & 2\nu+1<\kappa<2\nu+2,
\end{cases}
\end{align*}
\end{prp}

Proposition~\ref{prp:nonconsistent} shows that Assumptions~\ref{as:ED} and~\ref{as:DT} hold over a broad region, including $1/2<\kappa<2\nu+1+\gamma$ with some $\gamma<1/2$. Hence, by Theorem~\ref{thm:nonconsistent}, as the number of grid points $m$ increases, the independent-error model encourages over-clustering. This condition is naturally met under Gaussian noise process with the Mat\'ern kernel ($\nu>0$) and a signal function whose smoothness ranges from rougher than the noise process to slightly smoother, indicating that the independent-error assumption can be precarious in many data-analytic settings.

\subsection{Dependent error model}
We next consider the case where $\bCyp$ is general and admits error dependence. We impose two assumptions and show that the assumed model can achieve clustering performance equivalent to the true model. 

\begin{as}[Spectral closeness]\label{as:SC}
    For a constant $\beta>0$, it holds that
    \[\lambda_m(\bC_y),\, \lambda_m(\bCyp) \gtrsim m^{-\beta},\qquad \|\bC_y-\bCyp\|_{\mathrm{op}}=o(m^{-1-2\beta}).\]
\end{as}
\begin{as}[Sample size]\label{as:SS}
    The sample size $n$ satisfies $n = o(\exp (m) )$.
\end{as}
Assumption~\ref{as:SC} requires that the eigenvalues decay at most at a polynomial rate, a standard condition satisfied, for example, by Mat\'ern kernels.
In contrast, the assumption on the operator norm is strong. This condition implies that the absolute errors between the eigenvalues of $\bC_y$ and $\bCyp$ decay uniformly faster than $m^{-1-2\beta}$ from Weyl's inequality (see Lemma~\ref{lem:AbEigenIneq} in the supplementary material for details). Thus, $\bC_y$ and $\bCyp$ are required to be close in the spectral sense.
Assumption~\ref{as:SS} requires the sample size to grow sub-exponentially in the observation dimension. Under this condition, the following result continues to hold even when $n\to\infty$; in particular, the finite-$n$ case is covered.

\begin{thm}\label{thm:consistent}
Assume that $\bC_y$, $\bCyp$ and $\bC_\theta$ are positive definite kernels, and $\lambda_m(\bC_\theta) \gtrsim m^{-\beta}$ for a constant $\beta>0$. 
Under Assumptions~\ref{as:SC} and \ref{as:SS}, we have
\[
\frac{\pr^{\ast} \left(z_i=K_{-i}+1
         \mid \bz_{-i},
              \bth_{1:K_{-i}},
              \by_{1:n}\right)}
     {\pr^{\mathrm{true}} \left(z_i=K_{-i}+1
         \mid \bz_{-i},
              \bth_{1:K_{-i}},
              \by_{1:n}\right)}
 \xrightarrow{ p } 1,
\qquad m\to\infty.
\]
\end{thm}

Theorem~\ref{thm:consistent} shows that, under Assumptions~\ref{as:SC} and~\ref{as:SS} and as $m\to\infty$, the full conditional probability of creating a new cluster under the assumed model coincides with that of the true model. In other words, unlike in Theorem~\ref{thm:nonconsistent}, the assumed model does not generate new clusters more frequently than the true model within the MCMC algorithm, thereby preventing the inflation of clustering errors arising from the misspecification.

This theorem suggests that, in practice, it is crucial to construct $\bCyp$ to be spectrally close to the true $\bC_y$. As the true $\bC_y$ is unknown, we must estimate the sequence of eigenvalues, that is, the smoothness of the true function.

%% file: Proposal.tex
\section{Remedy: Correlated Errors and Scalable Implementation}

In this section, guided by Theorem~\ref{thm:consistent}, we seek a remedy that preserves theoretical guarantees while remaining computationally efficient. 
A straightforward solution is to model the error term with a Gaussian process, for instance using a Mat\'ern kernel. However, this could be computationally expensive when $m$ is large, since inverting the $m\times m$ covariance matrix typically requires $O(m^3)$ operations. Another option is to place an inverse-Wishart prior on the covariance matrix, under which the error covariance can be estimated at the $n^{-1/2}$ rate via the sample covariance~\citep{Koltchinskii2015Concentration}, thereby satisfying Assumption~\ref{as:SC}. Yet, it suffers from the same cubic computational burden as Gaussian processes and the estimation error could be considerable when the observation dimension $m$ exceeds the sample size $n$.

To reduce computational cost and provide stable estimation, we consider a parameter-reduction approach via banded matrices proposed in \cite{Lee2023post}.
A matrix $\bSigma=(\sigma_{ij}) \in\mathbb S_{++}^m$ is banded with bandwidth $r$ if $\sigma_{ij}=0$ whenever $|i-j|>r$ for all $i,j=1,\dots,m$. 
This structure reduces the memory cost from $O(m^2)$ to $O(mr)$, and banded Cholesky factorization enables efficient matrix inversion in $O(mr^2)$ operations~\citep{GolubVanLoan2013}.
For the true covariance matrix $\bC_y$, let $\bC_{\mathrm{approx}}$ be its banded-matrix approximation.
The following result justifies the use of $\bC_{\mathrm{approx}}$ in the correlated-error model when the bandwidth parameter is appropriately chosen. 

\begin{prp}\label{prp:band-approx}
If $C_y$ is the Mat\'ern kernel and the bandwidth is $r = \omega((1+2\beta)\log m)$, then we have $\|\bC_y-\bC_{\mathrm{approx}}\|_{\mathrm{op}} =o(m^{-1-2\beta})$.
\end{prp}

The above result shows that $\bC_{\rm approx}$ satisfies Assumption~\ref{as:SC} if we set the bandwidth parameter $r$ greater than $(1+2\beta)\log m$.
Since the computational cost associated with $\bC_{\rm approx}$ increases with $r$, the above proposition provides a suitable rate for $r$ to balance the computational cost and theoretical justification. 
Moreover, under Assumption~\ref{as:SS}, the estimation error of $\bC_{\rm approx}$ can be negligible compared to the approximation error in Assumption~\ref{as:SC} \citep[See e.g.][]{Lee2023post}.
Then, we adopt the Gaussian process model with $\bC_{\mathrm{approx}}$ as a practical remedy for functional clustering.

We implement the proposed approach in Gibbs sampler that maintains the cluster indicators, the cluster-specific mean vectors, and a banded working covariance matrix for the observational noise. 
At each step, observations are reassigned using the CRP full conditional
\[
\pr(z_i=k\mid\cdot)\propto
\begin{cases}
n_{-i,k}\,\phi_m(\by_i\mid\bth_k,\bC_{\rm approx}), & k=1,\ldots, K_{-i},\\
        \alpha\,\phi_m(\by_i\mid\bm 0,\,\bC_{\rm approx}+\bC_\theta), & k=K_{-i}+1,
\end{cases}
\]
where $\bC_\theta$ is a banded or non-banded version of the correlation matrix.
Given labels, each cluster mean admits a closed-form Gaussian update.
Mat\'ern smoothness and length-scale hyperparameters are updated via a Metropolis–Hastings algorithm~\citep{gelman2013bayesian}; scale parameters admit inverse-gamma updates. All linear algebra exploits banded Cholesky factorizations. Choosing the bandwidth $r$ to be on the order of $(1+2\beta)\log m$ mitigates over-generation of clusters and, as a result, avoids increases in the costs of CRP-based assignments and cluster-mean updates. In this sense, the proposed bandwidth selection may yield a shorter total runtime than the independent-error model, where the Gaussian-process computation itself is small. (See the following experiments.)

%% file: Simulation.tex
\section{Simulation}
We conduct simulation experiments to examine how the specification of the error distribution, in particular the independence assumption on discretized observations, affects posterior clustering under the functional mixture model introduced in Section~2. Our design varies the grid density, the correlation structure of the observation noise, and the prior on partitions. 

\subsection{Simulation setup}
We generate $n$ functional observations on a common grid $\{x_j\}_{j=1}^m \subset [0,1]$, where $m$ represents the grid density. Let $K^*=2$ be the true number of clusters and $z_i \in \{1,\ldots,K^*\}$ the latent membership. 
First, we independently sample $\theta_k(\cdot)$ from a Gaussian process ${\rm GP}\left(\eta(\cdot),\, C_\theta(\cdot,\cdot)\right)$ for $k=1,\ldots,K^*$ as in \eqref{mixture}.
Then, conditional on the cluster-specific mean functions $\{\theta_k(\cdot)\}_{k=1}^{K^*}$, the $i$th function is generated as $y_i(x) = \theta_{z_i}(x) + \varepsilon_i(x)$, $x \in \mathcal{X}=[0,1]$, where $\varepsilon_i(\cdot)$ follows a mean-zero Gaussian process. Discretizing at $\{x_j\}_{j=1}^m$ yields
\[
    \by_i=\bth_{z_i}+\beps_i,\qquad  \beps_i\sim \mathcal{N}\!\left(\mathbf{0},\,\bC_y \right),
\]
where $\by_i=(y_i(x_1),\ldots,y_i(x_m))^\top$ and $\bth_{k}=(\theta_k(x_1),\ldots,\theta_k(x_m))^\top$. We set $\eta(\cdot)\equiv 0$ and take $C_\theta$ to be the Gaussian kernel
$C_\theta(x,x')=\tau_\theta^2\exp\!\left\{-(x-x')^2/2\ell_\theta^2\right\}$, with $(\tau_\theta^2,\ell_\theta)=(1,\,0.15)$. Cluster sizes are balanced so that $n_k=n/K^*$ for $k=1,\ldots,K^*$. We fix $n=80$ and vary $m$ over moderate-to-large regimes; specifically, $m\in\{8,16,32,64\}$. The grid is equispaced on $[0,1]$.
To probe the impact of misspecifying the noise, we consider the following families for $\bC_y$:
\begin{enumerate}
\item \emph{Independent and identically distributed (IID)}: $\bC_y = \sigma_\varepsilon^2 \bm{I}_m$. This scenario serves as a baseline where the independence assumption is correctly specified.

\item \emph{Exponential kernel correlated (Exp)}: The covariance is governed by an exponential kernel, $(\bC_y)_{jk} = \sigma_\varepsilon^2 \exp(-|x_j - x_k|/\ell)$. We test two different correlation lengths, $\ell \in \{0.1, 1.0\}$. This structure represents a common form of smooth dependence.

\item \emph{Fractional Brownian motion correlated (fBm)}: The covariance structure is derived from a fractional Brownian motion process, defined by $(\bC_y)_{jk} = \sigma_\varepsilon^2(|x_j|^{2H} + |x_k|^{2H} - |x_j - x_k|^{2H})/2$. We consider two Hurst parameters, $H \in \{0.25, 0.5\}$. These values correspond to processes with rough sample paths ($H=0.25$) and standard Brownian motion ($H=0.5$), respectively.
\end{enumerate}
In other words, we generate noise under five different smoothness conditions: IID, Exp with two levels of smoothness, and fBm with two levels of smoothness. We fix the noise scale parameter $\sigma_\varepsilon^2$ at $0.05$ here and vary it in the Supplementary Materials.

We generate $50$ different datasets and analyze each of them using the following five Bayesian clustering procedures that differ in the partition prior and/or the working model for the error covariance:
\begin{enumerate}
\item \textbf{DP+IID}: This is the canonical Bayesian nonparametric clustering model. It pairs a Chinese Restaurant Process prior, ${\rm CRP}(\alpha)$, for the partition $(z_1,\ldots,z_n)$ with the independent-error model ($\bCyp=\sigma^2 \bm{I}_m$). 
\item \textbf{PY+IID}: This method replaces the Dirichlet process prior (inducing CRP) with a more flexible Pitman--Yor process prior, ${\rm PY}(\delta, \alpha)$, characterized by a discount parameter $\delta \in [0,1)$ and a concentration parameter $\alpha > -\delta$. The Pitman--Yor process can generate partitions with a power-law behavior in cluster sizes, offering more prior control over the number of clusters. The likelihood follows the independent-error model.
\item \textbf{DP+GP}: This approach substitutes $\bCyp$ in DP+IID by Mat\'ern kernel; namely the error process is modeled as a Gaussian process with a Mat\'ern kernel. The smoothness, length-scale and scale parameters are selected using MH algorithm.
\item \textbf{PY+GP}: This variant combines ${\rm PY}(\delta,\alpha)$ with a dependent-error working model. Specifically, the noise follows a Gaussian process with a Mat\'ern kernel, whose smoothness and length-scale are updated as in DP+GP. 
\item \textbf{band}: From DP+GP, we retain the CRP prior for the partition but replace a Mat\'ern kernel with a banded one, introduced in Section~4. Concretely, we construct a banded estimator $\widehat{\bC}_y$ of the noise covariance with bandwidth $r=\min (\lceil 3\log m\rceil, m-1)$ and use the assumed likelihood $\beps_i\sim\mathcal{N}(\mathbf{0},\widehat{\bC}_y)$.
\end{enumerate}
DP+IID model serves as a baseline, representing the standard approach which is deliberately misspecified under our Exp and fBm data-generating designs. PY+IID allows us to investigate whether a more flexible partition prior alone can mitigate the over-clustering issue when the observational error is misspecified.
All methods share the same Gaussian prior for $\bth_k$ as in~\eqref{mixture}, ensuring that differences arise from the treatment of the error structure and the partition prior.

Regarding the hyperparameters, we fix the partition parameters at $\alpha=1$ and $\delta=0.1$ (discount parameter used only in Pitman--Yor based methods) here and vary them in the Supplementary Material. For the assumed noise covariance in GP-based methods, we place a discrete uniform prior on the Mat\'ern smoothness
$\nu \in \{1/2,\ 3/2,\ 5/2,\ \infty\}$, where $\nu=\infty$ denotes the Gaussian kernel. Conditional on $\nu$, the length-scale parameters $(\ell_y,\ell_\mu)$ follow flat priors on the log-scale within bounds $\ell\in[0.01,10]$ tied to the grid diameter; in computation, proposals are made by random walks on $\log\ell$ with reflecting bounds. For the scale parameters of the noise and mean-process GPs, we adopt weakly informative inverse-gamma priors, $\tau_y^2,\tau_\mu^2 \sim {\rm Inv\text{-}Gamma}(1,1)$, which yield conjugate updates.

We summarize the posterior over partitions by the point estimate that minimizes the posterior expected variation of information \citep{wade2018bayesian,rastelli2018optimal}. Let $K_z$ denote the number of non-empty clusters induced by a draw of $Z$ and we report $E[K_z\mid \by_{1:n}]$. We evaluate clustering performance by the adjusted Rand index (ARI) \citep{hubert1985comparing}. Both metrics measure the agreement between the true and predicted cluster assignments, where higher values of ARI reflect better clustering accuracy. Furthermore, the accuracy of the estimated mean function for each sample is evaluated by the root mean squared error (RMSE):
\[
    {\rm RMSE}_\theta=\frac{1}{n}\sum_{i=1}^{n}\left\{\frac{1}{m}\sum_{j=1}^m\bigl(\hat{\theta}_{\hat{z}_i}(x_j)-\theta_{z_i}(x_j)\bigr)^2\right\}^{1/2}.
\]
All criteria are averaged over $50$ replicates for each $m$ and noise design.

Finally, all computations were implemented in Python~3.11. 
For the runtime results in Table~\ref{tab:computation_time}, experiments ran on macOS~14.1 (Apple silicon; 12-core CPU) with 36 GB unified memory.
We used 3{,}000 MCMC samples after 2{,}000 burn-in (no thinning). Package versions and sample code are provided in the Supplementary Material.

\subsection{Results}
First, in Figure~\ref{fig:number}, we examine the behavior of the number of clusters to justify the theoretical analysis in Section~3. 
Under the IID noise setting, where the working model is correctly specified, all methods accurately recover the true number of clusters ($K=2$), with the posterior distributions being tightly concentrated.
In contrast, when there is dependence among the noise terms (Exp and fBm settings), methods that assume IID noise (DP+IID and PY+IID) consistently and severely overestimate the number of clusters. The overestimation of DP (or CRP), proven in Theorem~\ref{thm:nonconsistent}, cannot be remedied by simply using more flexible partition priors (e.g. Pitman--Yor process). 
In contrast, methods that assume the error correlation (DP+GP, PY+GP and band) are robust to the underlying noise structure. They can adjust the kernel smoothness and hence consistently estimate the number of clusters to be near the true value of $2$ across all dimensions and noise types. This result supports Theorem~\ref{thm:consistent}, suggesting that as long as the working covariance is close to the true noise kernel, the number of clusters is not overestimated.

Because the overestimation of cluster counts directly translates into poor estimation performance of clusters and the underlying mean functions, we next investigate them through ARI and RMSE. From Figure~\ref{fig:ari}, for the correlated noise scenarios, the ARIs for the DP+IID and PY+IID methods are consistently smaller than $1$. Conversely, the DP+GP, PY+GP and band methods achieve values near 1, reflecting their success in recovering the true latent group structure.
Figure~\ref{fig:rmse} shows that the failure of the IID-based methods also propagated to the estimation of the cluster-mean functions. By failing to group similar curves, these methods could not leverage the ``borrowing strength'' phenomenon, where information is pooled across cluster members. Consequently, their RMSE for the mean functions is significantly higher than that of the GP-based methods. The latter, by correctly identifying the clusters, produced highly accurate estimates of the two true underlying mean functions.

\begin{figure}
    \includegraphics[width=\linewidth]{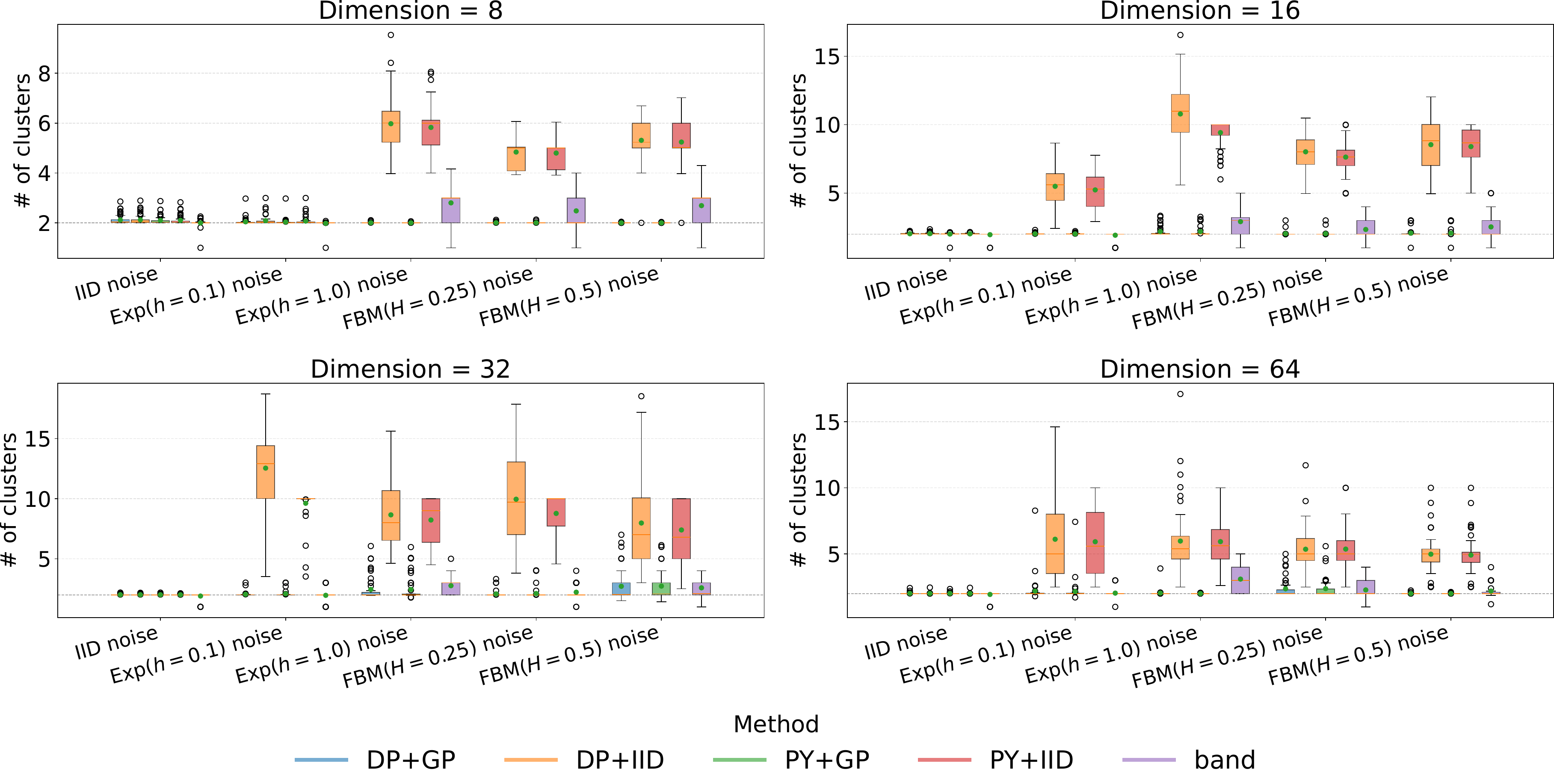}
    \caption{Posterior number of clusters for grid sizes $m\in\{8,16,32,64\}$ under five noise designs: independent errors; exponential kernel with correlation length $0.1$; exponential kernel with correlation length $1.0$; fractional Brownian motion with Hurst parameter $0.25$; fractional Brownian motion with Hurst parameter $0.5$. Methods: Dirichlet process with independent errors; Pitman--Yor process with independent errors; Dirichlet process with Gaussian process error (Mat\'ern kernel); Pitman--Yor process with Gaussian process error (Mat\'ern kernel); banded covariance model. The true number of clusters is $2$.}
    \label{fig:number}
\end{figure}

\begin{figure}
    \centering
    \includegraphics[width=\linewidth]{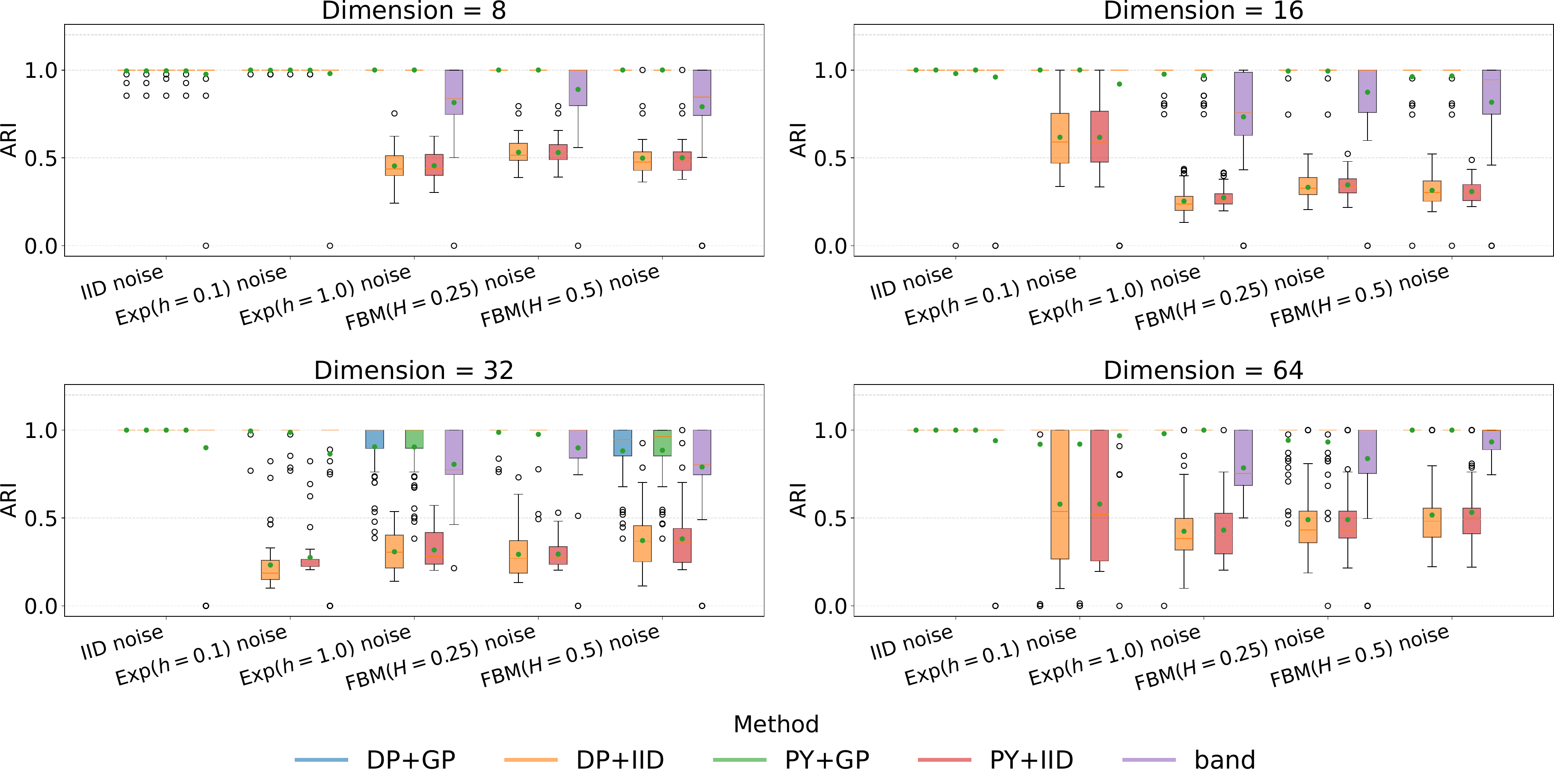}
    \caption{Adjusted Rand index (ARI) for grid sizes $m\in\{8,16,32,64\}$ across the same five noise designs and five clustering methods.}
    \label{fig:ari}
\end{figure}

\begin{figure}
    \centering
    \includegraphics[width=\linewidth]{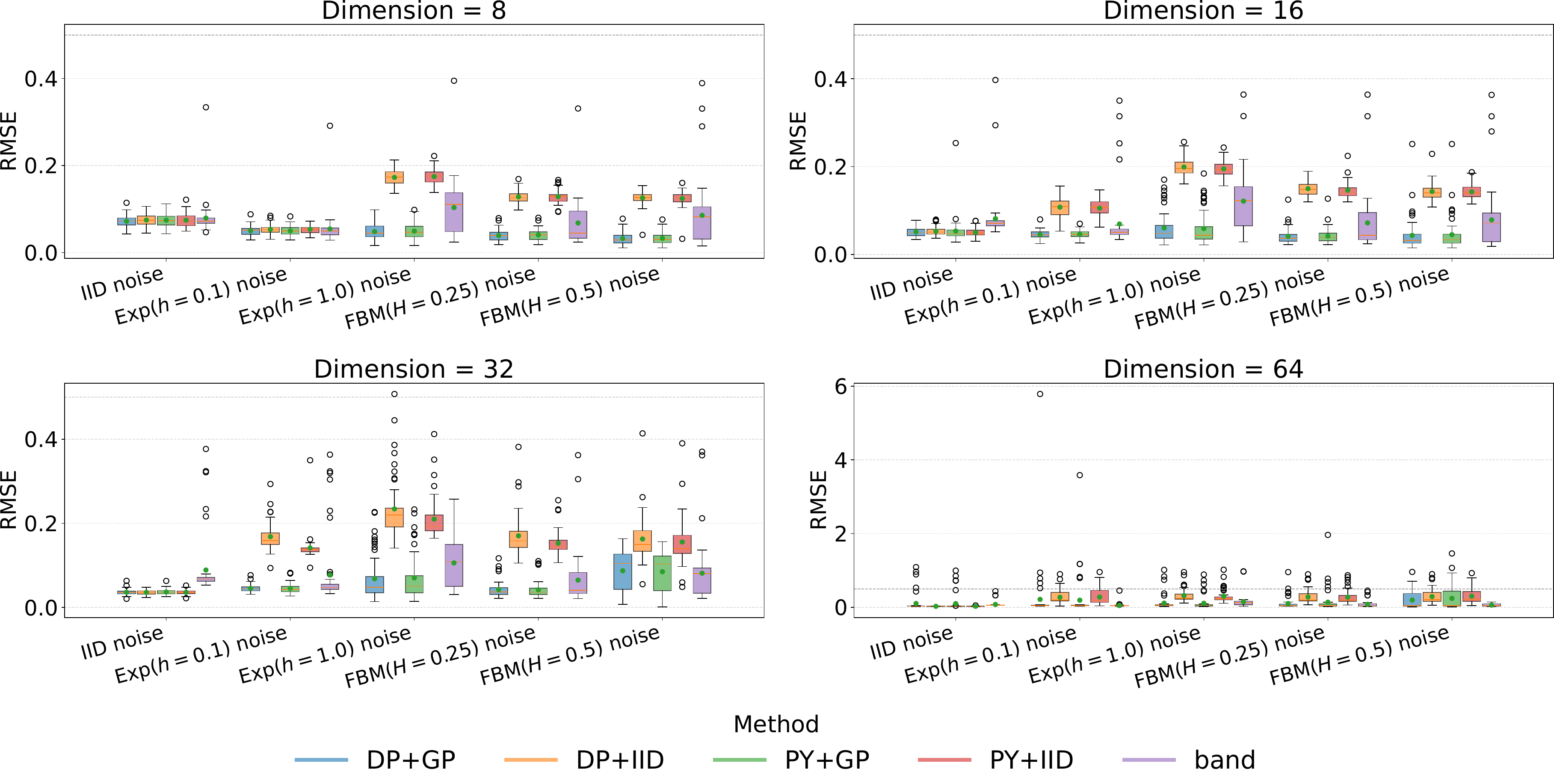}
    \caption{Root mean squared error (RMSE) of the estimated mean functions for grid sizes $m\in\{8,16,32,64\}$ across the five noise designs and five clustering methods.}
    \label{fig:rmse}
\end{figure}

We benchmark wall-clock time per MCMC run for each method across dimension $m\in\{8,16,32,64\}$, averaging over $50$ replicates. Table~\ref{tab:computation_time} reports the mean time (seconds) with standard deviations. For small grids ($m\le 16$), all four dense-kernel likelihoods (DP+GP, PY+GP, DP+IID, PY+IID) have similar runtime, while the banded likelihood incurs an overhead from estimating and applying $\widehat{\bC}_y$, making it slower at $m=8$ and $m=16$. As $m$ grows, the picture reverses: at $m=32$ and $m=64$ the band variant is the fastest by a wide margin (for example, $22.75$ seconds compared with approximately $31$ to $33$ seconds at $m=32$, and $30.70$ seconds compared with approximately $49$ to $53$ seconds at $m=64$). This is consistent with the computational benefit of exploiting banded structure in the working covariance, which scales more favorably than dense-kernel approaches as the grid densifies. Another notable point is that the IID model does not yield a substantially shorter total computation time than the GP-based models because it tends to overproduce clusters. Although computations involving covariance matrices should be considerably cheaper, generating many clusters implies that the set of candidates considered in CRP-based assignments grows and the cost of updating cluster means scales; consequently, the total computation time increases. Therefore, even with the goal of reducing overall runtime, the choice of bandwidth in the banded approach is crucial.

\begin{table}[]
\caption{Average computation time (in seconds) for different methods across varying dimension, averaged over $50$ data sets; standard deviations are in parentheses.}
\label{tab:computation_time}
\begin{center}
\begin{tabular}{lcccc}
\hline
\multirow{2}{*}{Method} & \multicolumn{4}{c}{Dimension ($m$)} \\
\cline{2-5}
 & 8 & 16 & 32 & 64 \\
\hline
DP+IID & 16.69 (0.54) & 18.49 (1.29) & 33.43 (3.83) & 49.58 (8.30) \\
PY+IID & 16.63 (0.53) & 18.35 (1.14) & 33.15 (3.31) & 49.53 (7.29) \\
DP+GP  & 16.47 (0.07) & 17.40 (0.27) & 31.31 (2.45) & 52.71 (5.40) \\
PY+GP  & 16.51 (0.08) & 17.46 (0.30) & 31.45 (2.53) & 53.16 (6.03) \\
band   & 19.02 (0.37) & 20.16 (0.60) & 22.75 (1.13) & 30.70 (2.22) \\
\hline
\end{tabular}
\end{center}
\end{table}

In summary, our simulation study clearly shows that assuming independent errors for densely observed functional data can be catastrophic for Bayesian functional clustering, leading to model failure across multiple performance indices. Appropriately modeling the error correlation structure resolves these issues with better accuracy and computation, even with a standard Dirichlet process prior.

%% file: Discussion.tex
\section{Discussion}

Our theoretical analysis primarily focuses on the behavior of the full conditional posterior probability of cluster assignment. This perspective provides valuable insights into how differences in error distributions affect clustering performance, and the proposed remedies based on the theory show promising results in numerical experiments.
However, it would be desirable to develop more direct theoretical evaluations of clustering performance. 
For example, establishing the consistency of the number of clusters would be a possible way, as considered in \cite{miller2013simple,miller2014inconsistency,10.1093/biomet/asac051} outside the functional data framework.
Hence, investigating whether partition distributions such as the Dirichlet process yield consistent clustering results under large samples when using correlated error distributions would be an interesting direction for future work.

%% file: Supplemet.tex
\setcounter{equation}{0}
\renewcommand{\theequation}{S\arabic{equation}}
\setcounter{section}{0}
\renewcommand{\thesection}{S\arabic{section}}
\setcounter{table}{0}
\renewcommand{\thetable}{S\arabic{table}}
\setcounter{figure}{0}
\renewcommand{\thefigure}{S\arabic{figure}}
\setcounter{lem}{0}
\renewcommand{\thelem}{S\arabic{lem}}

\setcounter{page}{1}

\begin{center}
{\LARGE {\bf 
Supplementary Materials of ``On Misspecified Error Distributions in Bayesian Functional Clustering: Consequences and Remedies''
}}
\end{center}

\bigskip
In the Supplementary Materials, we provide the details of the posterior computation algorithms and proofs of the main theorems. 

\section{Posterior Computation}

This section presents the complete Gibbs sampler that accompanies the full‑conditional derivations of the main text. The sampler combines a CRP update for the cluster indicators with conjugate updates for the remaining parameters, and incorporates a band-matrix post-processing step to ensure that the sampled error covariance $\bC_y$ remains computationally manageable. 
We provide pseudocode (Algorithm~1) for posterior computation. The remaining methods (DP+IID, PY+IID, DP+GP, PY+GP) can all be implemented by making only minor, local modifications to this pseudocode. Concretely: (i) to switch between a Dirichlet process prior and a Pitman--Yor process prior, replace the CRP seating probabilities in step~(a) by their PY counterparts (i.e., adjust the existing–cluster weights from $n_{-i,k}$ to $n_{-i,k}-\delta$ and the new–cluster weight from $\alpha$ to $\alpha+\delta K^{\setminus i}$); (ii) to move between IID and GP error models, change only the working covariance in the likelihood evaluations and linear solves, replacing $\bC_y=\sigma^2 \bm{I}_m$ by a kernel matrix (Mat\'ern or Gaussian) or vice versa; and (iii) to avoid the banded approximation, remove the operator $\mathrm{band}_r(\cdot)$. No other parts of the sampler need to be altered, so a single code path suffices with these modular switches for the partition prior and the error covariance.

\begin{algorithm}[]\label{algo}
\caption{CRP-Gibbs sampler with banded covariance}
\begin{algorithmic}[1]
\Require $\{ \by_i \}_{i=1}^n \subset \mathbb R^m$, iterations $T$,
         concentration $\alpha>-\delta$, discount $\delta\in[0,1)$,
         kernels $k_y(\,\cdot\,;\phi_y),\,k_\mu(\,\cdot\,;\phi_\mu)$,
         band width $r$,
         IG priors $(a_y,b_y)$ and $(a_\mu,b_\mu)$ for scales $\tau_y^2,\tau_\mu^2$
\Statex \textbf{Notation:}
  $\bC_y\equiv \mathrm{band}_r\!\big(k_y(\phi_y)\big)$,\;
  $\bC_\mu\equiv \mathrm{band}_r\!\big(k_\mu(\phi_\mu)\big)$.
\Statex \textbf{Initialisation:}
        $K\gets1$;\,
        $z_i\gets1$ $(i=1,\ldots,n)$;\,
        $\bm\theta_1\sim\mathcal N(\bm0,\bC_\theta)$;\,
        $\bC_y\gets\sigma^2\bm{I}_m$ with
        $\sigma^2\sim\operatorname{IG}(a_0,b_0)$
\For{$t=1$ {\bf to} $T$}
  \Statex\hspace*{-0.3em}\textbf{(a) Update cluster indicators $z_i$}
  \For{$i=1$ {\bf to} $n$}
    \State Remove $\by_i$ from its current cluster; update counts
           $n_{-i,k}$
    \For{each existing cluster $k=1,\dots,K^{\setminus i}$}
        \State
        $p_k\gets n_{-i,k}\,
                \mathcal N_m\!\bigl(\by_i\mid\bm\theta_k,\bC_y\bigr)$
    \EndFor
    \State $p_{\text{new}}\gets
           \alpha\,
           \mathcal N_m\!\bigl(\by_i\mid\bm0,\bC_y+\bC_\theta\bigr)$
    \State Draw $z_i$ from the Categorical distribution over
           $\{1,\dots,K^{\setminus i},\text{new}\}\propto
           \{p_1,\dots,p_{K^{\setminus i}},p_{\text{new}}\}$
    \If{$z_i$ is {\rm new}}
        \State $K\gets K+1$;\,
               $\bm\theta_{K}\sim\mathcal N(\bm0,\bC_\theta)$
    \EndIf
  \EndFor
  \Statex\hspace*{-0.3em}\textbf{(b) Update cluster means $\bm\theta_k$}
  \For{$k=1$ {\bf to} $K$}
     \State $n_k\gets\#\{i:z_i=k\}$,\;
           $\bm y_{\!k}^{\text{sum}}\gets\sum_{i:z_i=k}\by_i$
     \State $\bm{V}_k\gets\bigl(\bC_\theta^{-1}+n_k \bC_y^{-1}\bigr)^{-1}$, \; $\bm m_k\gets \bm{V}_k\,\bC_y^{-1}\bm y_{\!k}^{\text{sum}}$
     \State $\bm\theta_k\sim\mathcal N_m\bigl(\bm m_k,\bm{V}_k\bigr)$
  \EndFor
  \Statex\hspace*{-0.3em}\textbf{(c) Update scale parameters $\tau_y^2,\tau_\mu^2$}
  \State Residuals: for each $i$, $\br_i\gets \by_i-\bmu_{z_i}$
  \State Compute quadratic forms via banded solves:
         $Q_y \gets \tau_y^2 \sum_{i=1}^n \br_i^\top \bC_y^{-1}\br_i,\quad Q_\mu \gets \tau_\mu^2 \sum_{k=1}^K \bmu_k^\top \bC_\mu^{-1}\bmu_k$
  \State Sample $\tau_y^2 \sim \mathrm{IG}\!\big(a_y+\tfrac{nm}{2},\ b_y+\tfrac{Q_y}{2}\big)$ and $\tau_\mu^2 \sim \mathrm{IG}\!\big(a_\mu+\tfrac{Km}{2},b_\mu+\tfrac{Q_\mu}{2}\big)$
  
  \Statex\hspace*{-0.3em}\textbf{(d) Update kernel hyperparameters (MH)}
  \State Propose $\phi_y'\sim q_y(\cdot\mid\phi_y)$; form $\bC_y'\gets \mathrm{band}_r\!\big(k_y(\phi_y')\big)$
  \State Compute acceptance ratio: $\log A_y=\log p\big(\{ \by_i \}\mid \bmu_{1:K},\tau_y^2 \bC_y'\big) - \log p\big(\{ \by_i \}\mid \bmu_{1:K},\tau_y^2 \bC_y\big)+ \log p(\phi_y')-\log p(\phi_y)+\log q_y(\phi_y\!\mid\!\phi_y')-\log q_y(\phi_y'\!\mid\!\phi_y)$.
  \State Accept with probability $\min(1,e^{\log A_y})$.
  \State Repeat analogously for $\phi_\mu$ to update $\bC_\mu$ and caches.
  
  \Statex Store $\bigl(z_{1:n},\bm\theta_{1:K},\bC_y\bigr)$
\EndFor 
\end{algorithmic}
\end{algorithm}

When using banded approximations, the bandwidth should be linked to the spectral-closeness requirement in Assumption~(A1). In particular, choosing a bandwidth $r = \Omega((1+2\beta)\log m)$ ensures that the operator-norm approximation error satisfies $\|\bC_y - \mathrm{band}_r(\bC_y)\|_{\mathrm{op}} = o(m^{-1-2\beta})$, so the banding step itself respects the rate demanded by Assumption~(A1). Consequently, the theoretical guarantees used in the main text (notably the comparison in Theorem~\ref{thm:consistent}) continue to hold under this choice. In practice, the banded matrix may have a slightly negative smallest eigenvalue. To guarantee positive definiteness without altering the rates, apply a minimal diagonal shift: if $\lambda_{\min}$ denotes the smallest eigenvalue of $\mathrm{band}_r(\bC_y)$, set $\bC_y^{\mathrm{pd}} := \mathrm{band}_r(\bC_y) + (\varepsilon - \lambda_{\min}) \bm{I}_m$ with a tiny $\varepsilon > 0$. This correction enforces positive definiteness while leaving the asymptotic conditions unchanged.

\section{Additional Simulation}
In this section, we present two additional simulation studies to further substantiate the conclusions drawn in the main text. These experiments are designed to test the robustness of our findings under different noise levels and alternative prior specifications for the partition. For a more comprehensive evaluation of clustering performance, we also include the Purity Function (PF) \citep{manning2009introduction} in addition to the ARI.
Unless otherwise noted below, all generative and prior settings follow the baseline described in Section~4.1.

First, we examine the impact of a lower signal-to-noise ratio (SNR). To this end, we increase the error variance to $\sigma_y^2=0.1$, which doubles the value used in the main text. This setting provides a more challenging scenario for all methods, making it more difficult to distinguish the true underlying cluster structure from the observation noise. 
The results are illustrated in Figure~\ref{fig:all_experiment2}.
Reducing the SNR consistently degrades the performance of the IID-error models (DP+IID, PY+IID): they overestimate the number of clusters, ARI and PF decrease, and RMSE increases. The degradation intensifies as SNR decreases and persists across noise structures (short- and long-range exponential, fractional Brownian) and grid sizes. By contrast, the models assuming dependency (DP+GP, PY+GP, and the banded approximation) remain close to the true cluster count and achieve higher ARI and PF with lower RMSE. These findings support the theoretical claim that ignoring correlated noise raises the probability of spurious clusters, particularly under low SNR conditions, whereas modeling the error structure prevents such fragmentation.

\begin{sidewaysfigure} 
    \centering
    \includegraphics[width=0.49\linewidth]{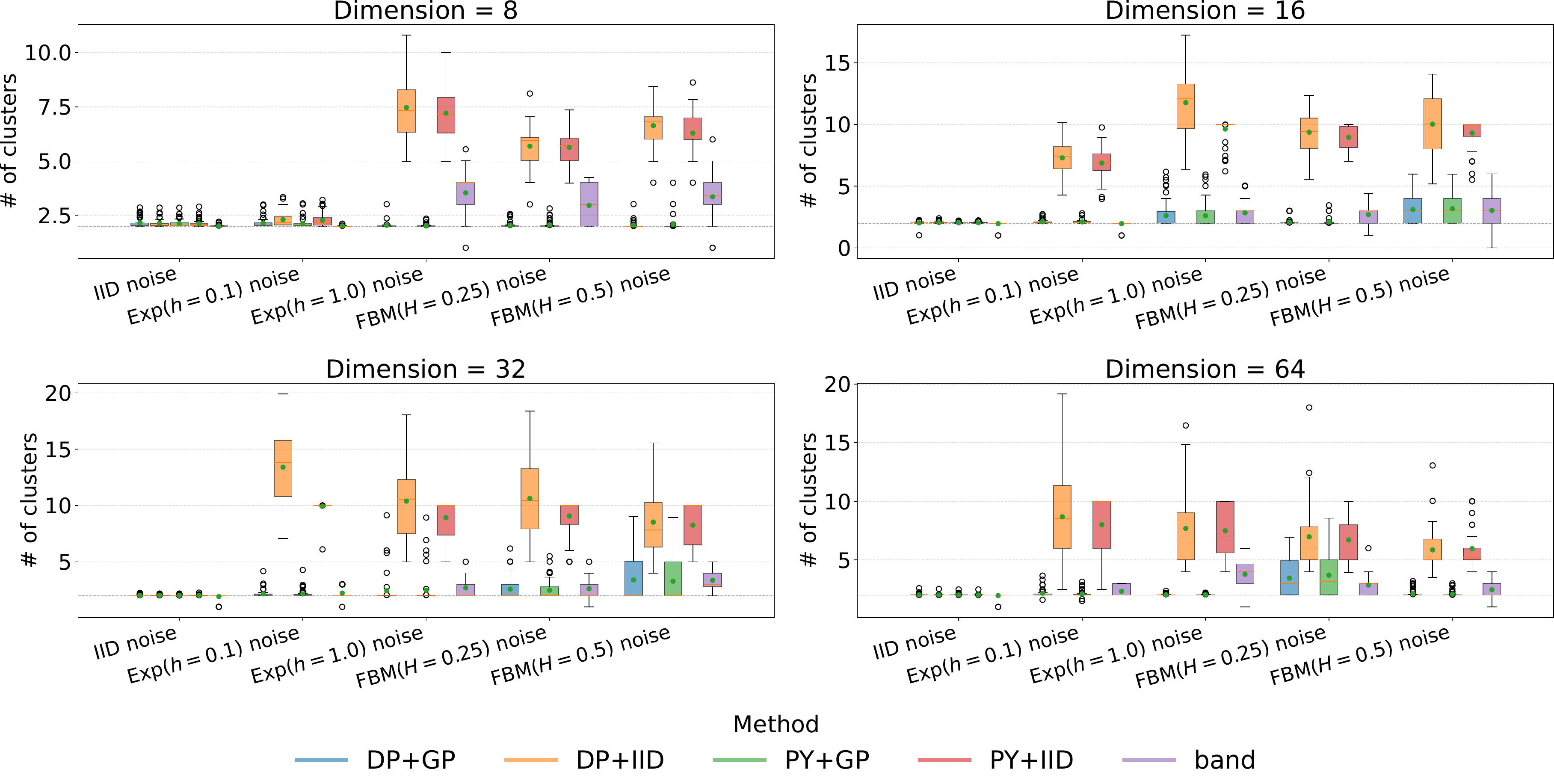}
    \includegraphics[width=0.49\linewidth]{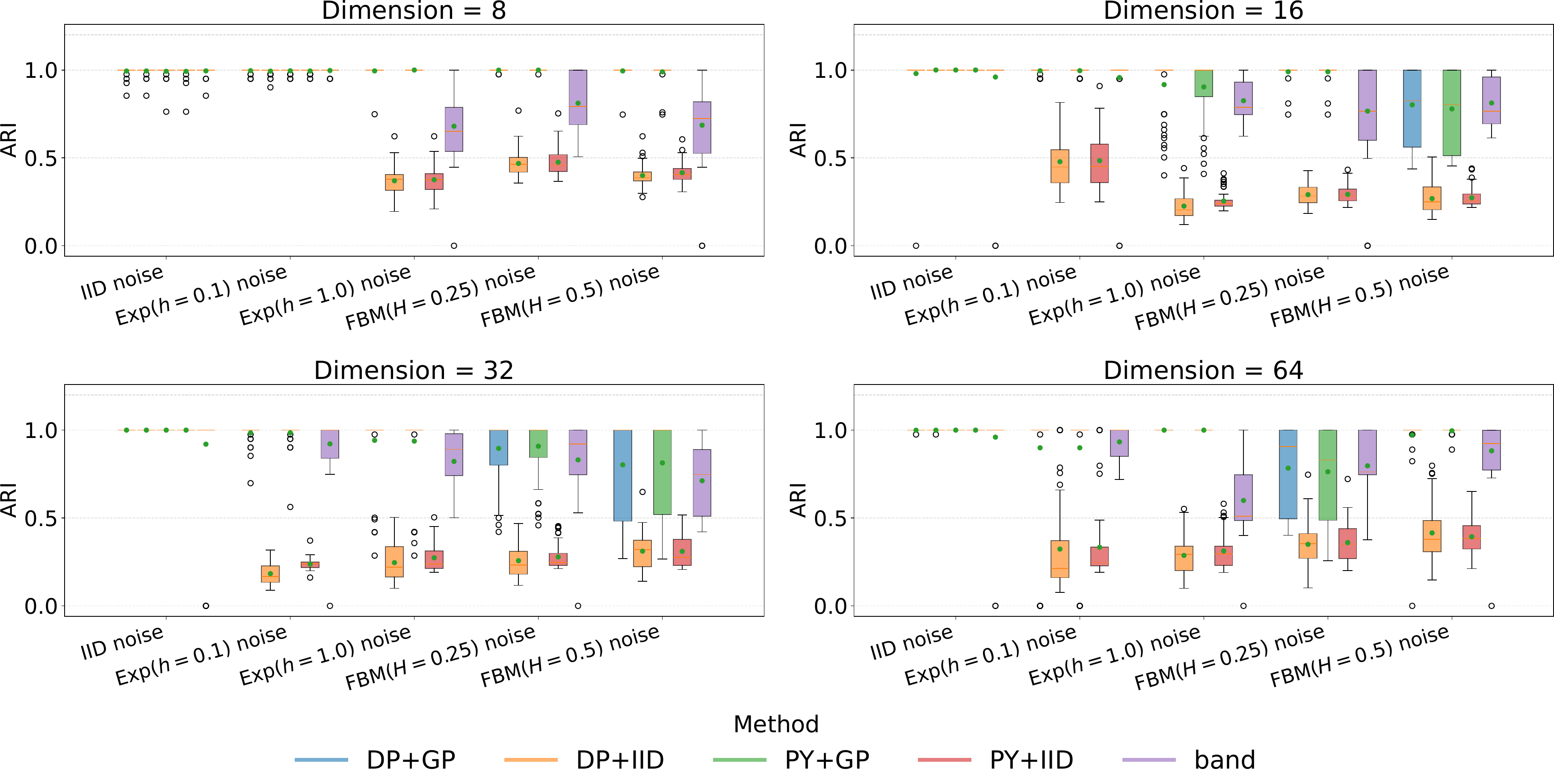}\\[1ex]
    \includegraphics[width=0.49\linewidth]{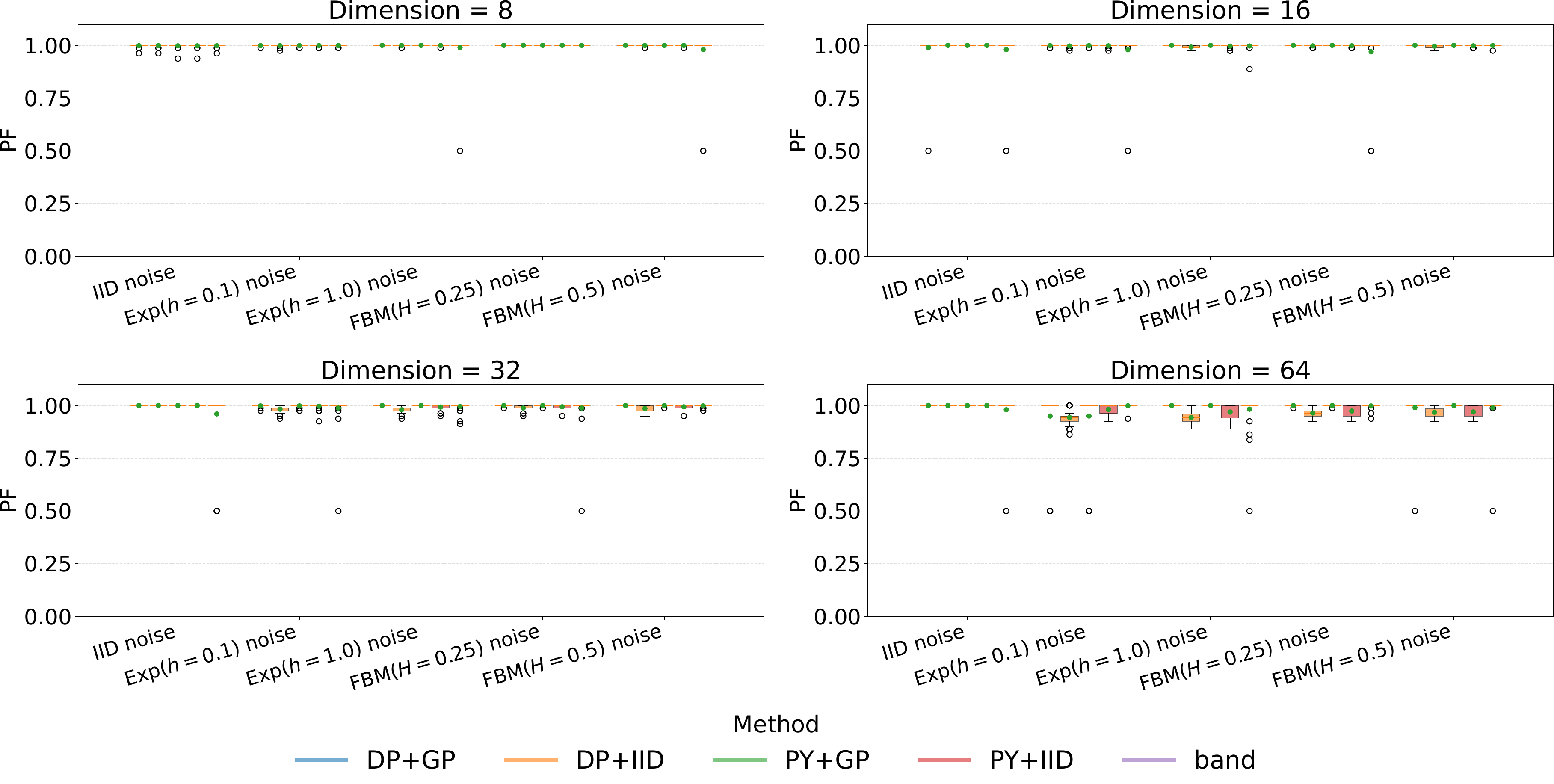}
    \includegraphics[width=0.49\linewidth]{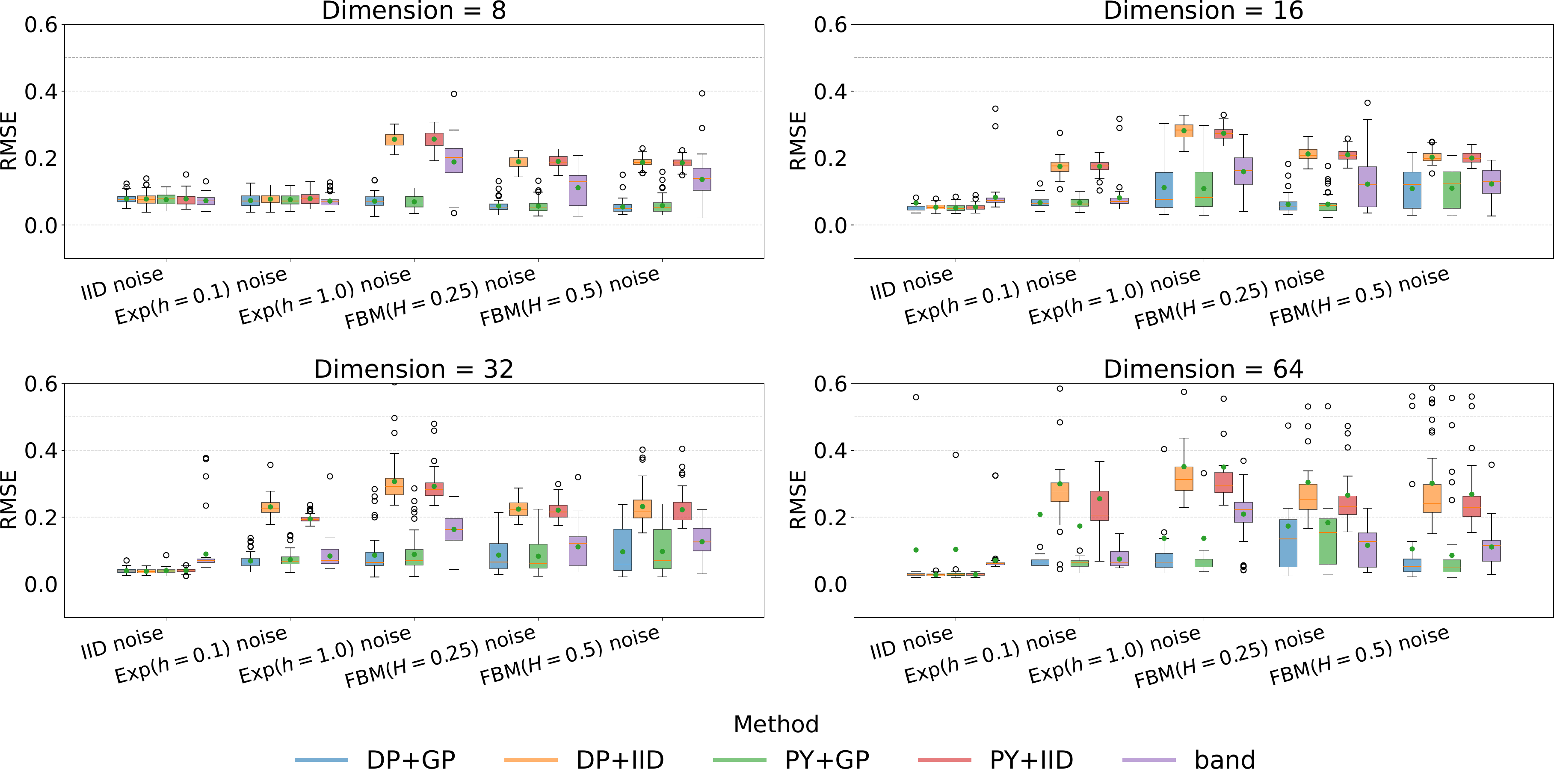}
    \caption{Posterior number of clusters, adjusted Rand index (ARI), purity function (PF), and root mean squared error (RMSE) for grid sizes $m\in\{8,16,32,64\}$ under five noise designs (independent errors; exponential kernel with correlation length $0.1$; exponential kernel with correlation length $1.0$; fractional Brownian motion with Hurst parameter $0.25$; fractional Brownian motion with Hurst parameter $0.5$) and five clustering methods (Dirichlet process with independent errors; Pitman--Yor process with independent errors; Dirichlet process with Gaussian process error using a Mat\'ern kernel; Pitman--Yor process with Gaussian process error using a Mat\'ern kernel; banded covariance model).}
    \label{fig:all_experiment2}
\end{sidewaysfigure}

Second, we investigate the influence of the partition prior by modifying the parameters of the Pitman--Yor process to $(\alpha,\delta)=(1.5,0.2)$. Compared to the parameters used in the main text $(\alpha,\delta)=(1.0,0.1)$, this choice of hyperparameters places more prior mass on partitions with existing clusters. The purpose of this experiment is to verify that the severe over-clustering observed in the IID-based models is a consequence of the error distributional misspecification, rather than being an artifact driven solely by the prior's tendency to favor fewer clusters.
Figure~\ref{fig:all_experiment3} shows the results.
Increasing the prior mass on existing clusters under the DP and PY processes does not prevent the observed over-splitting. Even with more conservative hyperparameters, IID-error models continue to overestimate the number of clusters and perform worse on ARI, PF, and RMSE, whereas the models assuming dependency remain stable and accurate. The consistent ordering across priors indicates that performance is driven primarily by the error-covariance specification rather than by the partition prior. The banded approximation closely matches the full GP, suggesting that the gains are robust to computational approximations. Thus, tuning the prior cannot fairly compensate for a misspecified noise model.

\begin{sidewaysfigure} 
    \centering
    \includegraphics[width=0.49\linewidth]{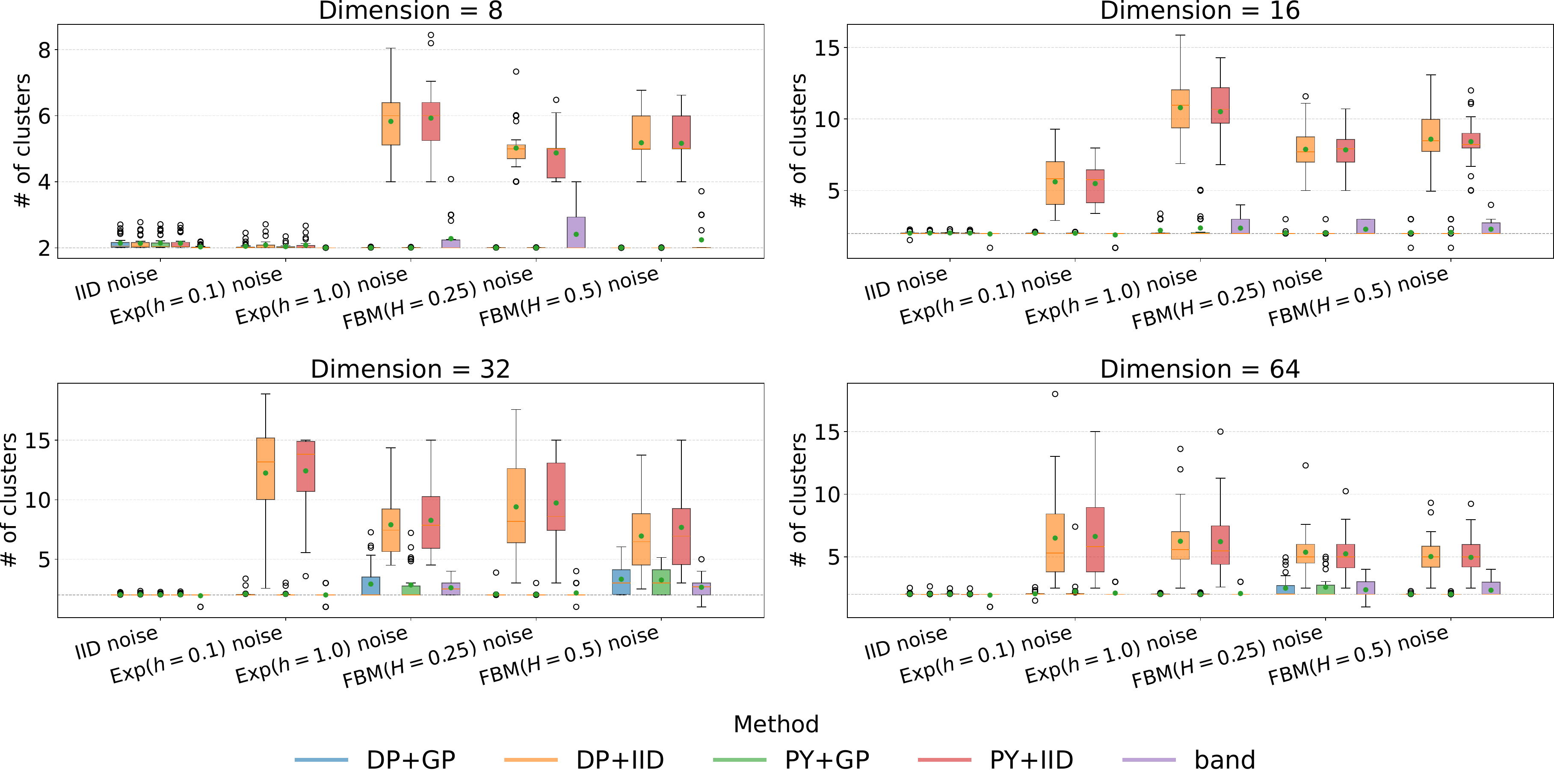}
    \includegraphics[width=0.49\linewidth]{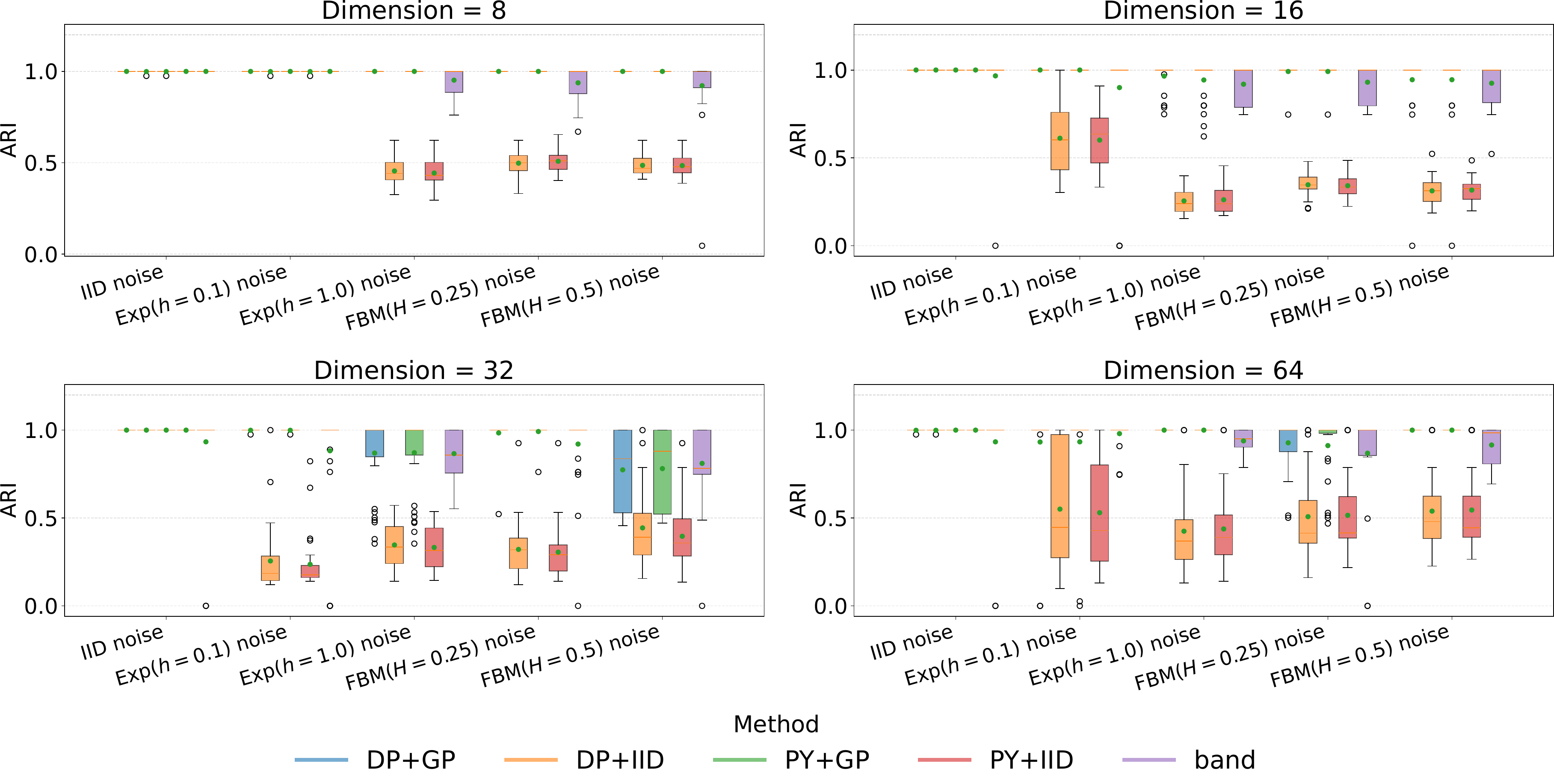}\\[1ex]
    \includegraphics[width=0.49\linewidth]{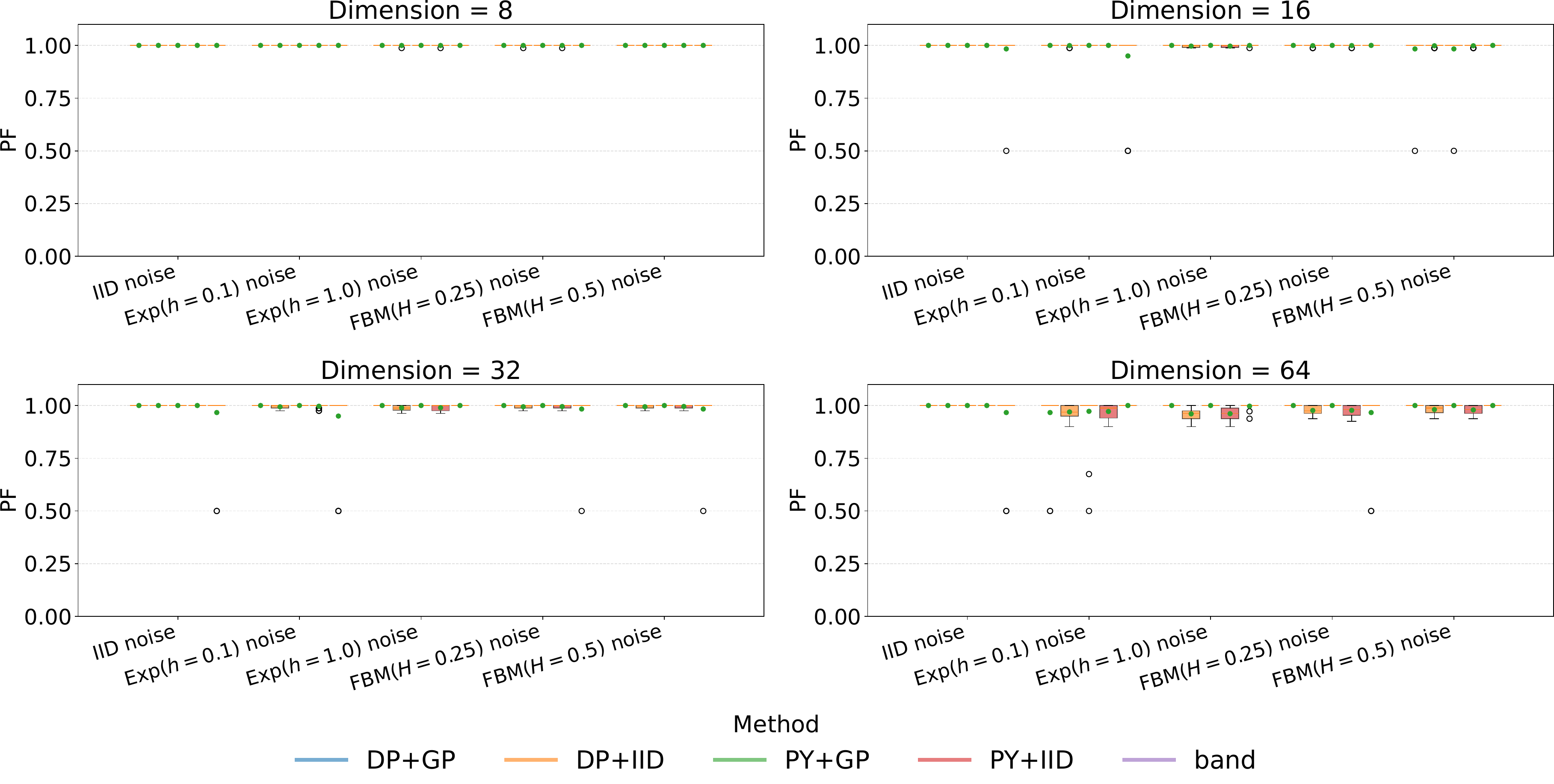}
    \includegraphics[width=0.49\linewidth]{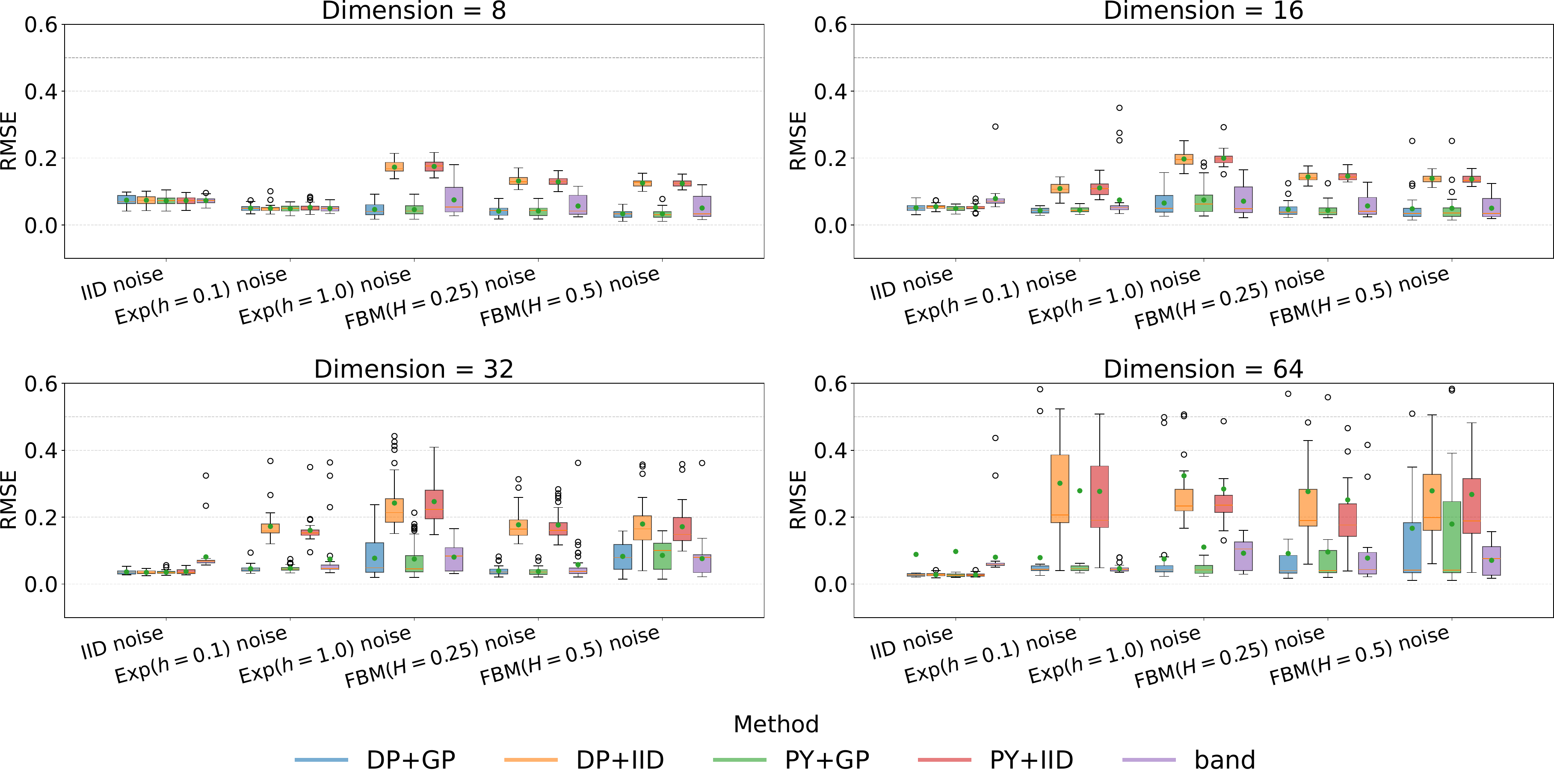}
    \caption{Posterior number of clusters, adjusted Rand index (ARI), purity function (PF), and root mean squared error (RMSE) for grid sizes $m\in\{8,16,32,64\}$ under five noise designs (independent errors; exponential kernel with correlation length $0.1$; exponential kernel with correlation length $1.0$; fractional Brownian motion with Hurst parameter $0.25$; fractional Brownian motion with Hurst parameter $0.5$) and five clustering methods (Dirichlet process with independent errors; Pitman--Yor process with independent errors; Dirichlet process with Gaussian process error using a Mat\'ern kernel; Pitman--Yor process with Gaussian process error using a Mat\'ern kernel; banded covariance model).}
    \label{fig:all_experiment3}
\end{sidewaysfigure}

\section{Proofs}

We provide the proofs of theorems and propositions.

\begin{proof}[Proof of Theorem 1]
The conditional probability ratio of interest is
\begin{align*}
     \frac{\pr^{\ast} \left(z_i=K_{-i}+1
         \mid \bz_{-i},
              \bth_{1:K_{-i}},
              \by_{1:n}\right)}
     {\pr^{true} \left(z_i=K_{-i}+1
         \mid \bz_{-i},
              \bth_{1:K_{-i}},
              \by_{1:n}\right)}
    = \frac{\alpha + \sum_{k=1}^{K_{-i}} n_{-i,k}\phi(\by_i\mid\bth_{k},\bC_y)/\phi(\bm{y_i}\mid\bm{0},\bSigma_2) }{\alpha + \sum_{k=1}^{K_{-i}} n_{-i,k}\phi(\by_i\mid\ \bth_{k},\sigma^2\bm{I}_m)/\phi(\bm{y_i}\mid\bm{0},\bSigma_1) },
\end{align*}
where $\bSigma_1 = \sigma^2\bm{I}_m + \bC_\theta, \bSigma_2 = \bC_y + \bC_\theta$. First, for every $k=1,\dots,K_{-i}$, we show that
\begin{align}\label{eq:first_step}
    \left. \frac{\phi(\by_i\mid\bth_{k},\bC_y)}{\phi(\bm{y_i}\mid\bm{0},\bSigma_2)}  \middle/ \frac{\phi(\by_i\mid\bth_{k},\sigma^2 \bm{I}_m)}{\phi(\bm{y_i}\mid\bm{0},\bSigma_1)} \xrightarrow{ p } \infty \right.,\qquad m\to\infty.
\end{align}
We denote the logarithm of the ratio of the numerators in \eqref{eq:first_step}:
\begin{align}
    L_{m}^{true} = \frac{1}{2}\log\left( \det \bSigma_2\right) - \frac{1}{2}\log\left( \det \bC_y\right)
     - \frac{1}{2}(\by_i - \bth_k)^\top \bC_y^{-1} (\by_i - \bth_k) + \frac{1}{2}\by_i^\top\bSigma_2^{-1}\by_i.  \label{eq:Lm_tr}
\end{align}
We decompose $\by_i$ as the sum of $\bth_k$ and $\bxi_i$ where $\bxi_i \sim \mathcal{N}(\bm{0}, \bC_y)$ and is independent of $\bth_k$. Then, $(\by_i - \bth_k)^\top \bC_y^{-1} (\by_i - \bth_k) = \bxi_i^\top \bC_y^{-1} \bxi_i$. The expectation and variance of $\bxi_i^\top \bC_y^{-1} \bxi_i$ are $m$ and $2m$, respectively. By the Chebyshev's inequality, for all $M > 0$ and $0 \leq \gamma < 1/2$,
\begin{align*}
    \pr\left(\left| \bxi_i^\top \bC_y^{-1}\bxi_i - m \right| \geq Mm^{1-\gamma}\right) &\leq 2M^{-2}m^{2\gamma-1}.
\end{align*}
Therefore, 
\begin{align}\label{eq:Lm_xi}
    &\bxi_i^\top \bC_y^{-1}\bxi_i = m + O_p(m^{1-\gamma}).
\end{align}
Also, $\by_i^\top \bSigma_2^{-1} \by_i = \bth_k^\top \bSigma_2^{-1} \bth_k + 2\bth_k^\top\bSigma_2^{-1}\bxi_i + \bxi_i^\top \bSigma_2^{-1} \bxi_i$. For the first term $\bth_k^\top \bSigma_2^{-1} \bth_k$, the expectation is $\E{\bth_k^\top \bSigma_2^{-1} \bth_k} = \tr(\bSigma_2^{-1}\bC_\bth)$. Here, we set $\bth_k = \bC_\theta^{1/2}\bze_k, \bS_\theta = \bC_\theta^{-1/2}\bC_y\bC_\theta^{-1/2}$, where $\bze_k \sim \mathcal{N}(\bm{0}, \bm{I}_m)$. Since $\bC_\theta$ and $\bC_y$ are symmetric positive definite matrix, $\bS_\theta$ is also symmetric positive definite, and $\bSigma_2$, $\bSigma_2^{-1}$ and $\bth_k^\top \bSigma_2^{-1} \bth_k$ can be written in terms of $\bS_\theta$ as follows: 
\[
    \bSigma_2 = \bC_\theta^{1/2}(\bm{I}_m + \bS_\theta)\bC_\theta^{1/2},\quad \bSigma_2^{-1} = \bC_\theta^{-1/2}(\bm{I}_m + \bS_\theta)^{-1}\bC_\theta^{-1/2},\quad \bth_k^\top \bSigma_2^{-1} \bth_k = \bze_k^{\top}(\bm{I}_m + \bS_\theta)^{-1}\bze_k.
\]
$\|\cdot\|_{\mathrm{F}}$ and $\|\cdot\|_{\mathrm{op}}$ denote the Frobenius norm and the operator norm, respectively. Then, 
\begin{align*}
    &\|(\bm{I}_m + \bS_\theta)^{-1}\|_{\mathrm{F}}^2 = \tr\left\{\left((\bm{I}_m + \bS_\theta)^{-1}\right)^\top (\bm{I}_m + \bS_\theta)^{-1}\right\} = \sum_{j=1}^{m}\left(\frac{1}{1 + \lambda_j(\bS_\theta)}\right)^2 \leq m,\\
    &\|(\bm{I}_m + \bS_\theta)^{-1}\|_{\mathrm{op}} = \max_{j=1,\dots,m}\frac{1}{1 + \lambda_j(\bS_\theta)} \leq 1.
\end{align*}
The $\zeta_j$ are independent and identically distributed as $\mathcal{N}(0,1)$, and there exists an positive constant $K$ such that $\|\zeta_j\|_{\psi_2} \leq K$, where $\|\cdot\|_{\psi_2}$ denote the sub-gaussian norm. By the Hanson-Wright inequality \citep{rudelson2013hanson}, for all $t > 0$,
\begin{align*}
    &\pr\left(\left| \bze_k^{\top}(\bm{I}_m + \bS_\theta)^{-1}\bze_k - \tr (\bSigma_2^{-1}\bC_\theta) \right| \geq t\right)\\
    &\leq \exp\left(-c \min\left\{ \frac{t^2}{K^4\|(\bm{I}_m + \bS_\theta)^{-1}\|_{\mathrm{F}}^2}, \frac{t}{K^2\|(\bm{I}_m + \bS_\theta)^{-1}\|_{\mathrm{op}}} \right\}\right)
    \leq \exp\left(-c \min\left\{ \frac{t^2m^{-1}}{K^4}, \frac{t}{K^2} \right\}\right).
\end{align*}
Hence,
\begin{align}\label{eq:Lm_theta}
    \pr\left(\left| \bze_k^{\top}(\bm{I}_m + \bS_\theta)^{-1}\bze_k -\tr (\bSigma_2^{-1}\bC_\theta) \right| \geq Mm^{1-\gamma}\right) &\leq \exp\left(-c \min\left\{ \frac{M^2m^{2-2\gamma}}{K^4}, \frac{Mm^{1-\gamma}}{K^2} \right\}\right),\notag\\
    \bth_k^\top\bSigma_2^{-1}\bth_k=\bze_k^{\top}(\bm{I}_m + \bS_\theta)^{-1}\bze_k &=  \tr(\bSigma_2^{-1}\bC_\theta) + O_p(m^{1-\gamma}).
\end{align}
For the second term $2\bth_k^\top\bSigma_2^{-1}\bxi_i$, the expectation is $0$, since this is the inner product of independent multivariate normal random variables with expectation $\bm{0}$. The variance is 
\begin{align*}
    \var{\bth_k^\top \bSigma_2^{-1} \bxi_i}
    &= \tr(\bSigma_2^{-1}\bC_\theta\bSigma_2^{-1}\bC_y)\\
    &= \tr(\bC_\theta^{-1/2}(\bm{I}_m + \bS_\theta)^{-1}\bC_\theta^{-1/2} \bC_\theta \bC_\theta^{-1/2}(\bm{I}_m + \bS_\theta)^{-1}\bC_\theta^{-1/2} \bC_y )\\
    &= \tr( (\bm{I}_m + \bS_\theta)^{-2}\bS_\theta )\\
    &= \tr( \bm{U}_{\bS_\theta}^\top(\bm{I}_m + \bm{\Lambda}_{\bS_\theta})^{-2}\bm{\Lambda}_{\bS_\theta} \bm{U}_{\bS_\theta} )\\
    &= \sum_{j=1}^{m}\frac{ \lambda_j(\bS_\theta) }{( 1 + \lambda_j(\bS_\theta) )^2},
\end{align*}
where $\bm{U}_{\bS_\theta}$ is an orthogonal matrix such that $\bS_\theta = \bm{U}_{\bS_\theta}^\top\bm{\Lambda}_{\bS_\theta}\bm{U}_{\bS_\theta}$, where $\bm{\Lambda}_{\bS_\theta} = \mathrm{diag}(\lambda_{1}(\bS_\theta),\dots,\lambda_m(\bS_\theta))$. Since a function $f(x) = x/(1 + x)^2$ on $\mathbb{R}_{>0}$ achieves its maximum $1/4$ at $x=1$, so $\var{\bth_k\bSigma_2^{-1}\bxi_i}$ is bounded above by $m/4$. By the Chebyshev's inequality,
\begin{align*}
    &\pr\left(\left| \bth_k^\top \bm{\Sigma}_2^{-1} \bxi_i \right| \geq Mm^{1-\gamma}\right) \leq \frac{M^{-2}m^{2\gamma-1}}{4}.
\end{align*}
Therefore, 
\begin{align}\label{eq:Lm_inner}
    \bth_k^\top \bm{\Sigma}_2^{-1} \bxi_i =  O_p(m^{1-\gamma}).
\end{align}
For third term $\bxi_i^\top\bSigma_2^{-1}\bxi_{i}$, the same result as $\bth_k^\top \bm{\Sigma}_2^{-1}\bth_k$
can be derived in exactly the same way by interchanging the roles of $\bC_y$ and $\bC_\theta$:
\begin{align}\label{eq:Lm_xi_Sigma2}
    \bxi_i^\top \bSigma_2^{-1} \bxi_i = \mathrm{tr}(\bSigma_2^{-1}\bC_y) + O_p(m^{1-\gamma}).
\end{align}
From \eqref{eq:Lm_tr} to \eqref{eq:Lm_xi_Sigma2}, we have
\begin{align}\label{eq:Lm_order}
     L_m^{true}
     &= \frac{1}{2}\log\left( \det \bSigma_2\right) - \frac{1}{2}\log\left( \det \bC_y\right) + O_p(m^{1-\gamma}).
\end{align}
We denote the following log-ratio of the assumed model in \eqref{eq:first_step} by $L_m^{*}$:
\begin{align*}
    L_{m}^{\ast} &= \log \frac{\phi(\by_i\mid\bth_{k},\sigma^2\bm{I}_m)}{\phi(\by_i\mid\bm{0},\bSigma_1)}\\
    &= \frac{1}{2}\sum_{j=1}^{m}\log\left( 1 + \frac{\lambda_j(\bC_\theta)}{\sigma^2} \right) - \frac{1}{2\sigma^2}\|\bm{y}_i - \bth_k\|^2 + \frac{1}{2}\by_i^\top\bSigma_1^{-1}\by_i.
\end{align*}
The following result can be derived in exactly the same as $L_m^{true}$:
\[
    L_{m}^{\ast} = \frac{1}{2}\sum_{j=1}^{m}\log\left( 1 + \frac{\lambda_j(\bC_\theta)}{\sigma^2} \right) - \frac{1}{2\sigma^2}\tr(\bC_y) + \frac12 \tr\left\{\bSigma_1^{-1}\bSigma_2\right\} + O_p(m^{1-\gamma}).
\]
By Assumption~1,~2,
\begin{align*}
    L_m^{true} - L_m^{*} 
    &= \omega(m^{1-\gamma}) + o(m^{1/2}) + O_p(m^{1-\gamma}),
\end{align*}
which validates \eqref{eq:first_step}. Finally, from \eqref{eq:first_step} and \eqref{eq:Lm_order}, we obtain, as $m\to\infty$,
\begin{align*}
    &\frac{\alpha + \sum_{k=1}^{K_{-i}} n_{-i,k}\phi(\bm{y}_{i}\mid\bth_{k},\sigma^2\bm{I}_m)/\phi(\by_i\mid\bm{0},\bSigma_1) }{\alpha + \sum_{k=1}^{K_{-i}} n_{-i,k} \phi(\by_{i}\mid\bth_{k},\bC_y)/\phi(\by_i\mid\bm{0},\bSigma_2) }\\
    &=\frac{\alpha}{\alpha + \sum_{k=1}^{K_{-i}} n_{-i,k} \phi(\by_{i}\mid\bth_{k},\bC_y)/\phi(\by_i\mid\bm{0},\bSigma_2)} + \sum_{k=1}^{K_{-i}}\frac{ n_{-i,k}\phi(\bm{y}_{i}\mid\bth_{k},\sigma^2\bm{I}_m)/\phi(\by_i\mid\bm{0},\bSigma_1)}{\alpha + \sum_{k} n_{-i,k} \phi(\by_{i}\mid\bth_{k},\bC_y)/\phi(\by_i\mid\bm{0},\bSigma_2)}\\
    &\leq \frac{\alpha  \phi(\by_i\mid\bm{0},\bSigma_2)}{n_{-i,1} \phi(\by_{i}\mid\bth_{k},\bC_y)} + \sum_{k=1}^{K_{-i}}\frac{\phi(\bm{y}_{i}\mid\bth_{k},\sigma^2\bm{I}_m)/\phi(\by_i\mid\bm{0},\bSigma_1)}{\phi(\by_{i}\mid\bth_{k},\bC_y)/\phi(\by_i\mid\bm{0},\bSigma_2)} \xrightarrow{ p } 0.
\end{align*}

\end{proof}

\begin{proof}[Proof of Proposition 1]
From the Stirling approximation $\log n! = n\log n-n+O(\log n)$, we have 
\[
(\det \bC_\theta/\det \bC_y)^{1/m}
     =(m!)^{\beta/m}
     =\exp\left(\beta\log m+O(1)\right)
     =cm^{\beta}(1+o(1)),
\]
where $\beta:=2\nu+1-\kappa$. By Lemma~1, 
\[
L\ge \log\left(1+cm^{\beta}\right)
     =\begin{cases}
        \beta m\log m+O(m), & \beta>0,\\
        m\log(1+c)+O(1),   & \beta=0,\\
        c m^{1+\beta}+o \left(m^{1+\beta}\right), & -1<\beta<0.
      \end{cases}
\]
Hence divergence occurs if $1+\beta>0( \Leftrightarrow \kappa<2\nu+2)$
with the rates.
\end{proof}

\begin{proof}[Proof of Theorem 2]
Write $\bSigma_3:=\bCyp+\bC_\theta$ and $\bSigma_2=\bC_{y}+\bC_\theta$.
Then, 
\[
\log \left(\frac{\phi(\bm{y_i}\mid\bm{0},\bSigma_3)}{\phi(\bm{y_i}\mid\bm{0},\bSigma_2)}\right)
 =\frac12\left\{\log|\bSigma_2|-\log|\bSigma_3|\right\}
 +\frac12\by_i^{\top}
     \left\{(\bSigma_3)^{-1}-(\bSigma_2)^{-1}\right\}\by_i.
\]
By Lemma~\ref{lem:log_det}, the determinant difference is $o(1)$.
By Assumption~3, $\|(\bSigma_3)^{-1}-(\bSigma_2)^{-1}\|_{\mathrm{op}} = \left\|\bSigma_2^{-1} (\bC_{y } - \bC_{y'}) \bSigma_3^{-1}\right\|_{\mathrm{op}} = o( m^{2\beta} m^{-1-2\beta}) = o(m^{-1})$.
Since $\|\by_i\|^2=O_p(m)$ for Gaussian $\by_i$,
$\by_i^{\top}\left(\bSigma_3^{-1}-\bSigma_2^{-1}\right)\by_i=o_p(1)$, and hence
\[
\frac{\phi(\bm{y_i}\mid\bm{0},\bSigma_3)}{\phi(\bm{y_i}\mid\bm{0},\bSigma_2)}=1+o_p(1). \label{eq:ratio32}
\]

For each existing cluster $k\le K_{- i}$, define
\[
\log
   \frac{\phi(\by_i\mid\bth_k,\bC_{y'})}
        {\phi(\by_i\mid\bth_k,\bC_{y})}
 =\frac12\left\{\log(\det \bC_y)-\log(\det \bC_{y'})\right\}
 +\frac12(\by_i-\bth_k)^{\top}
     \left\{(\bC_{y'})^{-1}-(\bC_{y})^{-1}\right\}
     (\by_i-\bth_k).
\]
From Lemma~\ref{lem:log_det}, $\log(\det \bC_y)-\log(\det\bC_{y'})=o(1)$.
By Assumption~3, $\|\bC_{y }^{-1}-\bC_{y'}^{-1}\|_{\mathrm{op}} = \left\|\bC_{y}^{-1} (\bC_{y} - \bC_{y'}) \bC_{y'}^{-1}\right\|_{\mathrm{op}} = o( m^{2\beta} m^{-1-2\beta}) = o(m^{-1})$. By Lemma~\ref{lem:ell2bound}, setting $\bm{\xi}_j\overset{iid}{\sim} \mathcal{N}(\bm{0},\Sigma)$, we have
\[
   \sup_{1\le k \le K_{-i}} \|\by_i -\bth_k\|^2 \le \sup_{1\le j \le n} \|\bxi_j\|^2 = O_p (m +2\sqrt{m\log n} +\log n),
\]
which is $O(m)$ from Assumption~4.
Thus, even if $K_{- i}$ may diverge,
\[
\frac{\phi(\by_i\mid\bth_k,\bC_{y'})}{\phi(\by_i\mid\bth_k,\bC_{y})} = 1 + o_p(1),
\]
uniformly in $k=1,\dots ,K_{- i}$.

Finally we evaluate
\begin{align*}
\frac{\pr^{\ast}(z_i=K_{- i}+1\mid\cdot)}
      {\pr^{\mathrm{true}}(z_i=K_{- i}+1\mid\cdot)}
&=\frac{\phi(\by_i\mid\bm{0},\bSigma_3)}{\phi(\by_i\mid\bm{0},\bSigma_2)}\cdot \frac{\alpha \phi(\by_i\mid\bm{0},\bSigma_2) + \sum_{k=1}^{K^{\setminus i}} n_{-i,k}\phi(\by_i\mid\bth_k,\bC_y) }{\alpha \phi(\by_i\mid\bm{0},\bSigma_3) + \sum_{k=1}^{K^{\setminus i}} n_{-i,k}\phi(\by_i\mid\bth_k,\bCyp) } \\
&=\frac{\phi(\by_i\mid\bm{0},\bSigma_3)}{\phi(\by_i\mid\bm{0},\bSigma_2)}
     \frac{\alpha+\sum_{k=1}^{K^{\setminus i}} n_{-i,k} \tfrac{\phi(\by_i\mid\bth_k,\bC_y)}{\phi(\by_i\mid\bm{0},\bSigma_2)}}
          {\alpha \tfrac{\phi(\by_i\mid\bm{0},\bSigma_3)}{\phi(\by_i\mid\bm{0},\bSigma_2)} +\sum_{k=1}^{K^{\setminus i}} n_{-i,k} \frac{\phi(\by_i\mid\bth_k,\bCyp)}
        {\phi(\by_i\mid\bth_k,\bC_y)} \frac{\phi(\by_i\mid\bth_k,\bC_y)}{\phi(\by_i\mid\bm{0},\bSigma_2)} }
\end{align*}
Recall that $\phi(\by_i\mid\bm{0},\bSigma_3) / \phi(\by_i\mid\bm{0},\bSigma_2)=1+o_p(1)$ and $\phi(\by_i\mid\bth_k,\bCyp)/\phi(\by_i\mid\bth_k,\bC_{y}) = 1+ o_p(1)$. Therefore, even if $n$ and $K^{\setminus i}$ diverge, $\pr^{\ast}(z_i=K^{\setminus i}+1\mid\cdot)/\pr^{\mathrm{true}}(z_i=K^{\setminus i}+1\mid\cdot)\overset{p}{\longrightarrow}1$.
\end{proof}

\begin{proof}[Proof of Proposition~2]
Define $\bC_y$ as the true covariance matrix, $\bC_{approx}$ as its band-matrix approx and $\hat{\bC}_y$ as an estimator of $\bC_{approx}$ by \cite{Lee2023post}.
For simplicity, let us assume the true kernel is Mat\'ern. Note that covariance is
\[  
    C_y(|i-j|)=\sigma^2 k_{|i-j|}, \qquad k_{|i-j|} \approx \left(|i-j|\right)^{\nu -\tfrac12} e^{-|i-j|}   \qquad (|i-j|\to\infty).
\]
and consider approximating it to a band matrix of the form:
\[
    C_{approx,ij} = C_{y}(|i-j|) I\{|i-j|\le r\}.
\]
Then, 
\begin{align*}
    \|\bC_y - \bC_{approx} \|_{\mathrm{op}} &\le \max_i \sum_{j:|j-i|>r} |C_y(|i-j|)|\\
    &\le 2 \sum_{h=r}^{\infty} C_y(h) \\
    &\le 2 \sum_{h=r}^{\infty} h^{\nu -\tfrac12} e^{-h}\\
    &\le 2 \int_{h=r}^{\infty}  h^{\nu -\tfrac12} e^{-h} dh  \\
    &\le 2 r^{ \nu - \tfrac12} e^{-r} \frac{ r }{ r-\nu + \tfrac12 }. 
\end{align*}
Here, the last inequality follows from Lemma~\ref{lem:band}. 
Then, the result follows if $r = \omega((1+2\beta)\log m)$.
\end{proof}

\section{Technical Lemmas}
\begin{lem}
Let $\bC_y,\bC_\theta\in\mathbb S_{++}^m$ be real, symmetric, positive--definite matrices and set
\[
\bS_y  =  \bC_y^{-1/2} \bC_\theta \bC_y^{-1/2}, \qquad 
L  :=  \log\left\{\det(\bm{I}_m + \bS_y)\right\}.
\]
Then the following lower bounds hold.
\begin{align}
L \ge  m \log \left(1+\left(\det \bC_\theta/\det \bC_y\right)^{1/m}\right).
\label{eq:minkowski}
\end{align}
\end{lem}

\begin{proof}
\[
L
=\log\det \left(\bm{I}_m+\bS_y\right)
=\log\frac{\det(\bC_y+\bC_\theta)}{\det \bC_y}.
\tag{$\ast$}
\]
By Minkowski determinant inequality,  
\[
\left(\det(\bC_y+\bC_\theta)\right)^{1/m}
 \ge 
(\det \bC_y)^{1/m}+(\det \bC_\theta)^{1/m}.
\]
Insert this inequality into $(\ast)$ and rearrange to obtain~\eqref{eq:minkowski}.
\end{proof}

\begin{lem}\label{lem:AbEigenIneq}
    Let $\bm{A}$ and $\bm{B}$ be $m \times m$ Hermitian matrices. Then,
    \[
    \left|\lambda_j(\bm{A}) - \lambda_j(\bm{B})\right| \leq \| \bm{A}-\bm{B} \|_{\mathrm{op}},\quad j=1,\dots,m.
    \]
\end{lem}
\begin{proof}
  By Weyl's inequality (Theorem I\hspace{-1.2pt}I\hspace{-1.2pt}I.2.1 or Corollary I\hspace{-1.2pt}I\hspace{-1.2pt}I.2.2 in \cite{bhatia1997matrix}), for each $j=1,\dots,m$,
  \begin{gather*}
      \lambda_j(\bm{B}) + \lambda_m(\bm{A}-\bm{B}) \leq \lambda_j(\bm{A}) = \lambda_j(\bm{B}+(\bm{A}-\bm{B})) \leq \lambda_j(\bm{B}) + \lambda_1(\bm{A}-\bm{B})\\
      -\| \bm{A}-\bm{B}\|_{\mathrm{op}} \leq \lambda_m(\bm{A}-\bm{B}) \leq \lambda_j(\bm{A}) - \lambda_j(\bm{B}) \leq \lambda_1(\bm{A}-\bm{B}) \leq \| \bm{A}-\bm{B} \|_{\mathrm{op}}\\
      | \lambda_j(\bm{A}) - \lambda_j(\bm{B}) | \leq \| \bm{A}-\bm{B} \|_{\mathrm{op}}
  \end{gather*}
\end{proof}

\begin{lem}\label{lem:log_det}
Fix $\beta > 0$ and, for $\bC_y, \bCyp, \bC_\theta \in \mathbb{S}_{++}^{m}$, assume that
\[
    \| \bC_y - \bCyp \|_{\mathrm{op}} = o(m^{-1-2\beta}),\qquad \lambda_m(\bC_y), \lambda_m(\bC_\theta) \gtrsim m^{-\beta}.
\]
Then, the following holds:
\begin{align*}
    &| \log(\det \bC_y) - \log(\det \bCyp) | = o(m^{-\beta}),\\
    &| \log\{\det (\bC_y + \bC_\theta)\} - \log\{\det (\bCyp + \bC_\theta)\} | = o(m^{-\beta})
\end{align*}
\end{lem}
\begin{proof}
By Lemma \ref{lem:AbEigenIneq},
\begin{gather*}
    \lambda_m(\bm{B}) \geq \lambda_m(\bm{A}) - \| \bm{A}-\bm{B} \|_{\mathrm{op}}\gtrsim m^{-\beta}.
\end{gather*}
Since $\bC_y, \bCyp \in \mathbb{S}_{++}^{m}$,
\begin{align*}
    \left| \log(\det \bC_y) - \log(\det \bCyp) \right| \leq \sum_{j=1}^{m}| \log\lambda_j(\bC_y) - \log\lambda_j(\bCyp) |.
\end{align*}
By the mean value theorem, for each $j=1,\dots,m$, there exists $c_j\in[\underline{\lambda}_j, \overline{\lambda}_j]$, where $\underline{\lambda}_j = \min\{\lambda_j(\bC_y), \lambda_j(\bCyp)\},\overline{\lambda}_j=\max\{\lambda_j(\bC_y),\lambda_j(\bCyp) \}$ such that
\begin{align*}
    \left| \log\{\lambda_j(\bC_y)\} - \log\{\lambda_j(\bCyp)\} \right| 
    &= c_j^{-1}| \lambda_j(\bC_y) - \lambda_j(\bCyp) |.
\end{align*}
Therefore,
\begin{align*}
\sum_{j=1}^{m}| \log\lambda_j(\bC_y) - \log\lambda_j(\bCyp) | &\leq \sum_{j=1}^{m}\underline{\lambda}_j^{-1}|\lambda_j(\bC_y) - \lambda_j(\bCyp) |\\
&\leq \| \bC_y-\bCyp \|_{\mathrm{op}}\sum_{j=1}^{m}\underline{\lambda}_j^{-1} \\
&\leq \| \bC_y-\bCyp \|_{\mathrm{op}}m\underline{\lambda}_m^{-1} \\
&\lesssim m^{1+\beta}\| \bC_y-\bCyp \|_{\mathrm{op}} = o(m^{-\beta}).
\end{align*}
\end{proof}

\begin{lem}\label{lem:ell2bound}
Assume $\|\bm{\Sigma}\|_{\mathrm{op}}$ is finite constant. For $\by \sim \mathcal N(\bm{0},\bSigma)$, we have $\|\by\|^2 \le \lambda_1(\bSigma)m+O_p(m^{1/2}) = O_p(m)$ as $m\to \infty$. For $n$ i.i.d. copies $\{\by_i\}_{i=1}^n$ of $\by$, $\pr(\sup_{1\le i\le n}\|\bSigma^{-1/2} \by_i\|^2 \ge m + 2\sqrt{2m\log n} + 4\log n )\le 1/n$.
\end{lem}
\begin{proof}
    The first assertion follows from the definition of $\|\cdot\|_{\mathrm{op}}$ and $\chi_2^{m} = m + O_p(m^{1/2})$. The second assertion follows from the Laurent--Massart $\chi^2$ bound and union bound.
\end{proof}

\begin{lem}\label{lem:band}
Let $\nu>1/2$ and $b$ satisfy $b+1>\nu-1/2$. Then
\[
   \int_{b+1}^{\infty} h^{ \nu-\frac12} e^{-h} dh
    \le 
   (b+1)^{\nu-\frac12} e^{-(b+1)} 
   \frac{ b+1 }{ b+1-\left(\nu-\tfrac12\right)}.
\]
\end{lem}

\begin{proof}
Set $x_0:=b+1$ and $s:=\nu-\tfrac12+1$; thus $s>1$ and $x_0>s-1$ by hypothesis.
Define
\[
   I := \int_{x_0}^{\infty}h^{ s-1} e^{-h} dh .
\]

Integrating by parts with 
$u=h^{s-1}$ and $dv=e^{-h}dh$ gives
\[
   I
   = \left[ - h^{s-1}e^{-h} \right]_{x_0}^{\infty}
     + (s-1)\int_{x_0}^{\infty} h^{ s-2}e^{-h} dh
   = x_0^{s-1}e^{-x_0} 
     + (s-1) \int_{x_0}^{\infty} h^{ s-2}e^{-h} dh .
\]
For $h\ge x_0$ we have $h^{ s-2}\le x_0^{-1}h^{ s-1}$, so
\[
   \int_{x_0}^{\infty} h^{ s-2}e^{-h} dh
    \le 
   x_0^{-1}\int_{x_0}^{\infty} h^{ s-1}e^{-h} dh
   = x_0^{-1} I .
\]
Hence
\[
   I 
   \le x_0^{s-1}e^{-x_0} 
       + (s-1)x_0^{-1}I .
\]
Rearranging,
\[
   I\left(1-\tfrac{s-1}{x_0}\right)
    \le  x_0^{s-1}e^{-x_0},
   \qquad\text{and}\qquad
   1-\tfrac{s-1}{x_0}>0
   \ \ (\text{since }x_0>s-1).
\]
Therefore
\[
   I
    \le 
   x_0^{s-1}e^{-x_0} 
   \frac{x_0}{x_0-(s-1)}
   =
   (b+1)^{\nu-\frac12}e^{-(b+1)}
   \frac{ b+1 }{ b+1-\left(\nu-\tfrac12\right)} ,
\]
which is exactly the desired inequality.
\end{proof}